\documentclass[10pt,journal,compsoc]{IEEEtran}

\usepackage{amsmath} 
\usepackage{amsfonts} 
\usepackage{amssymb}
\usepackage{amsthm}
\usepackage[utf8]{inputenc}
\usepackage{subfigure}
\usepackage{caption}
\usepackage{bbm}
\usepackage{cuted}
\usepackage[pdftex]{graphicx} 

\ifCLASSOPTIONcompsoc
  \usepackage[nocompress]{cite}
\else
  \usepackage{cite}
\fi
\ifCLASSINFOpdf
\else
\fi
\DeclareMathOperator*{\argmin}{argmin}

\newtheorem{defi}{Definition}
\newtheorem{prop}{Proposition}
\newtheorem{cor}{Corollary}
\newtheorem{theo}{Theorem}
\newtheorem{lem}{Lemma}
\newtheorem{NR}{Numerical Result}

\begin{document}
%
\title{Analysis of Static Cellular Cooperation between Mutually Nearest Neighboring Nodes}
%
%
%
%

\author{Luis~David~\'Alvarez~Corrales,
        Anastasios~Giovanidis,~\IEEEmembership{Member,~IEEE,}
        Philippe~Martins,~\IEEEmembership{Senior~Member,~IEEE,}
        and~Laurent~Decreusefond
\IEEEcompsocitemizethanks{\IEEEcompsocthanksitem Luis David \'Alvarez Corrales conducted this research while at T\'el\'ecom ParisTech, 23 avenue d'Italie, 75013, Paris, France.\protect\\
E-mail: luis.alvarez-corrales@telecom-paristech.fr
\IEEEcompsocthanksitem Anastasios Giovanidis is with the CNRS. He conducted this research while affiliated with T\'el\'ecom ParisTech, 23 avenue d'Italie, 75013, Paris, France.  \protect\\
E-mail: anastasios.giovanidis@telecom-paristech.fr

He is now affiliated with University Pierre et Marie Curie, CNRS-LIP6.
\protect\\
E-mail: anastasios.giovanidis@lip6.fr
\IEEEcompsocthanksitem Philippe Martins is with T\'el\'ecom ParisTech, 23 avenue d'Italie, 75013, Paris, France.
\protect\\
E-mail: philippe.martins@telecom-paristech.fr
\IEEEcompsocthanksitem Laurent Decreusefond is with T\'el\'ecom ParisTech, 23 avenue d'Italie, 75013, Paris, France.
\protect\\
E-mail: laurent.decreusefond@telecom-paristech.fr}
}

\IEEEtitleabstractindextext{%
\begin{abstract}
Cooperation in cellular networks is a promising scheme to improve
system performance. Existing works consider that a user dynamically chooses the stations
that cooperate for his/her service, but such assumption often has practical limitations. Instead,
cooperation groups can be predefined and static, with nodes linked by fixed infrastructure. To
analyze such a potential network, we propose a grouping method based on node
proximity. With the Mutually Nearest Neighbour Relation, we allow
the formation of singles and pairs of nodes. Given an initial topology for the stations, two new
point processes are defined, one for the singles and one for the pairs. We derive structural characteristics for these processes and analyse the resulting interference fields. When the node positions follow a Poisson Point Process (PPP) the processes of singles and pairs are not Poisson. However, the performance of the original model can be approximated by the superposition of two PPPs. This allows the derivation of exact expressions for the coverage probability. Numerical evaluation shows coverage gains from different signal cooperation that can reach up to $15\%$ compared to the standard noncooperative coverage. The analysis is general and can be applied to any type of cooperation in pairs of transmitting nodes.
\end{abstract}

\begin{IEEEkeywords}
Cooperation; Static groups; Poisson cellular network; Thinning; Interference; Poisson superposition.
\end{IEEEkeywords}}

\maketitle

\IEEEdisplaynontitleabstractindextext

%
\IEEEpeerreviewmaketitle

\IEEEraisesectionheading{\section{Introduction}\label{sec:introduction}}

\IEEEPARstart{C}{ooperation} between wireless nodes, such as cellular base stations (BSs) is receiving in recent years a lot of attention. It is considered as a way to reduce intercell interference in future cellular networks and consequently improve network capacity. It is particularly beneficial for users located at the cell-edge, where significant $\mathrm{SINR}$ gains can be achieved in the downlink. In the wireless literature, there is a considerable amount of research on the topic, which relates to the concept of CoMP \cite{CoorMulConIrmer,TheRolSmallCellsJungV}, Network MIMO \cite{CapMeasuJungnick2008,GesMultMIMO2010,GioA012012}, or C-RAN 
\cite{CloudRANChecko2015,DynResAllocLyaz2016}. It is also expected to play a significant role due to the coming densification of networks with HetNets \cite{DhillonBest12,OptSmallCellsChecko}. The various strategies proposed differ in the number of cooperating nodes, the type of signal cooperation, the amount of information exchange, and the way groups (clusters) are formed. 

Recent studies analyse such cooperative networks with Stochastic Geometry as the main analytic tool \cite{BacBlaVol1}. Modeling the position of wireless nodes via a Point Process gives the possibility to include the impact of irregularity of BS locations on the users' performance (e.g. $\mathrm{SINR}$, throughput, delay). Furthermore, the gains from cooperation can be quantified in a systematic way, so there is no need to test each different instance of the network topology by simulations. Closed formulas are very important for an operator that wants to plan and deploy an infrastructure with cooperation functionality, because these can provide intuition on the relative influence of various design parameters.

\subsection{Related Work}

There are important results available for BS cooperation in wireless networks. In \cite{BacAStoGeo2015}, Baccelli and Giovanidis analyse the case where BSs are modeled by a Poisson Point Process (PPP) and each user-terminal triggers the cooperation of its two closest BSs for its service. The authors show coverage improvements and an increase of the coverage cell. In \cite{NigCoordMul2014}, Nigam et al consider larger size of clusters, showing that BS cooperation is more beneficial for the worst-case user. The $\mathrm{SINR}$ experienced by a typical user when served by the $K$ strongest BSs is also investigated by Blaszczyszyn and Keeler in \cite{BlasStuSINTFact2015}, where the authors derive tractable integral expressions of the coverage probability for general fading by the use of factorial moment measures. An analysis of a similar problem with the use of Laplace Transforms (LT) is provided by Tanbourgi et al in \cite{TanTracMod2014}. Sakr and Hossain propose in \cite{LocAwaSakrHoss} a scheme between BSs in different tiers for downlink CoMP. Outside the Stochastic Geometry framework, we find \cite{GioA012012} and \cite{PapaADynmClust2008}. In \cite{PapaADynmClust2008}, Papadogiannis et al propose a dynamic clustering algorithm incorporating multi-cell cooperative processing. All the above works assume that a user-terminal dynamically selects the set of stations that cooperate for its service, which changes the cluster formation for every different configuration of users. This is difficult to be applied in practice. 

Other works propose to group BSs \textit{in a static way}, so that the clusters are a-priori defined and do not change over time. The appropriate static clustering should result in considerable performance benefits for the users, with a cost-effective infrastructure. In favour of the static grouping approach are Akoum and Heath \cite{AkouIntCoor2013}, who randomly group BSs around virtual centres; Park et al \cite{ParkColBSs}, who form clusters by using edge-coloring for a graph drawn by Delaunay triangulation; Huang et al \cite{AStocGemMultCellHuang2011}, who cluster BSs using a hexagonal lattice, and Guo et al who analyse in \cite{GuoSPGP2014} the coverage benefits of cooperating pairs modeled by a Gauss-Poisson point process \cite{StocGeoWirHaen2012}. The existing static clustering models either group BSs in a random way \cite{AkouIntCoor2013}, or they randomly generate additional cluster nodes around a cluster center \cite{GuoSPGP2014,AfshFundClustCent}, which is translated in the physical world into installing randomly new nodes in the existing infrastructure. A more appropriate analysis should have a map of existing BS locations as the starting point, and from this define in a systematic way cooperation groups. The criterion for grouping should be based on node proximity, in order to limit the negative influence of first-order interference.

\subsection{Mutually Nearest Neighbor cooperation}

Consider a fixed deployment of single antenna BSs on the plane. As argued above, we wish to organize these BSs (or atoms) into \textit{static cooperative groups}, with possibly different sizes. These groups must be mutually disjoint and their union should exhaust the whole set of BSs. Additionally, the groups must be invariable in size and elements with respect to the random parameters of the telecommunication network (e.g. fading, shadowing, or user positions). Hence, we look for a criterion that aims at \textit{network-defined, static clusters} as opposed to the user-driven selection of other works. 

For this reason, we will propose rules that depend only on geometry: An atom takes part in a group, based solely on its relative distance to the rest of the atoms. Geometry is related to the pathloss factor of the channel gain, so it encompasses important aspects that influence signal power.

The specific grouping criterion (for static geometric clusters) that we propose in this work is the \textit{Mutually Nearest Neighbor Relation (MNNR)}. 
The main idea is that two BSs belong to the same group if one of the two is the nearest neighbor of the other. The MNNR is the keystone that allows us to construct static clusters of singles and pairs. It is inspired by a model studied by Häggström and Meester \cite{HaggNNHS1996}, the \textit{Nearest Neighbor Model (NNM)}, and further analysed in \cite{DescChainLilyModDalLast2005,Kozakova06,DalStoSto99,GioAnalInt2015,AlvCovGains2016}, where each atom connects to its geometrically \textit{Nearest Neighbor} by an unidirected edge. Although we will consider here groups of size at most 2, the NNM can allow for an extension of our approach to include larger groups. This is part of our ongoing research. 
 
\subsection{Contributions}
This paper provides the following contributions:

\begin{itemize}
\item We introduce the MNNR, a grouping method for BSs whose positions are modeled by a stationary point process $\Phi$ (Section \ref{SecII}). Most results are derived when $\Phi$ is chosen to be a Poisson Point Process (PPP).
\item We analyse wireless networks with two types of clusters (groups): Single nodes, that do not cooperate, and pairs of nodes that cooperate with each other (Section \ref{SecII}).
\item From the dependent thinning determined by the MNNR, we construct two point processes $\Phi^{(1)}$ and $\Phi^{(2)}$, the processes of singles and pairs, respectively. Structural properties of both are provided: (a) the average proportion of atoms from $\Phi$ that belong to $\Phi^{(1)}$ and $\Phi^{(2)}$, (b) the average proportion of Voronoi surface related to each one of them, (c) their respective Palm measures, as well as (d) properties concerning repulsion/attraction (Section \ref{SecII}).   
\item Our analysis is done in a general sense, without restricting ourselves to specific cooperating signal schemes (Section \ref{SecIII}). Altogether, we provide the analytic tools that evaluate various strategies for transmitter cooperation/coordination, as those in \cite{SharStratKerr2013,TransCoopKerret2013,CoorMulConIrmer,TheRolSmallCellsJungV,GesMultMIMO2010}.
\item We provide an analysis of the interference generated by the processes $\Phi^{(1)}$ and $\Phi^{(2)}$, and derive explicit expressions for the corresponding expected values, along with a methodology to obtain their Laplace Transfrom (LT) (Section \ref{SecIV}).   
\item Based on the structural characteristics of the singles and the cooperative pairs, we introduce an approximate model: the superposition of two independent PPPs. Using this, a complete analysis of the coverage probability is provided, for two different scenarios of user-to-BS association (Section \ref{SecV}).
\item In Section \ref{SecVI} the analytic formulas are validated through simulations and the gains of static nearest neighbor grouping are quantified. Section \ref{SecVII} presents some pros and cons of the model. The final conclusions are drawn in Section \ref{SecVIII}.
\end{itemize}

\subsection{Notation}

Let all random elements be defined on a common probability space $(\Omega,\mathcal{F},\mathbb{P})$, and let $\mathbb{E}$ denote the expectation under $\mathbb{P}$. 



Let $\Phi=\{\phi\}$ be a point process, with values in $\mathbb{R}^2$, where every $\phi$ represents an element of the space of simple, and locally finite configurations of $\mathbb{R}^2$ points \cite{KalRanMea}.  

If an atom $x$ and a configuration $\phi$ are fixed, $\phi\cup \{x\}$ denotes the simple and locally finite configuration containing all the elements from $\phi$ plus $x$ in the case where $x$ does not belong to $\phi$, otherwise it just denotes $\phi$. In the same fashion, $\phi\backslash \{x\}$ denotes the simple and locally finite configuration containing all the elements from $\phi$ without the point $\{x\}$ in the case where $x$ actually belongs to $\phi$, otherwise it represents just $\phi$. 


The Euclidean distance and the Euclidean surface on $\mathbb{R}^2$ are denoted by $\| \cdot \|$ and $\mathcal{S}( \cdot )$, respectively.  

For $x\in \mathbb{R}^2$ and a $A$ closed subset of $\mathbb{R}^2$, we denote 
\begin{equation*}
d(x,A):=\inf_{y\in A} \|x-y\|,
\end{equation*}
the distance from $x$ to $A$. 

Finally, for $x\in \mathbb{R}^2$ and $r>0$, let 
\begin{equation*}
B(x,r):=\{y\in \mathbb{R}^2 \ | \ \|x-y\|<r\}.
\end{equation*}

\section{The Mutually Nearest Neighbor model}
\label{SecII}
 
\subsection{Singles and pairs}
\label{SecIIB}

For two different atoms $x$ and $y$ in the configuration $\phi$, we say that $x$ is in \textit{Nearest Neighbor Relation (NNR)} with $y$ (with respect to $\phi$) if
\begin{equation*}
y=\argmin_{z\in \phi \backslash \{x\}} \|x-z\|, 
\end{equation*}
and we write $x \stackrel{\phi}{\rightarrow} y$. When the atom $x$ is not in NNR with $y$, we write $x \stackrel{\phi}{\not\rightarrow} y$. 

Henceforth, we will  only consider stationary point processes $\Phi=\{\phi\}$ whose realisations fulfill the \textit{uniqueness} of the nearest neighbor a.s. Note however that this condition does not generally hold. For example, within the finite configuration $\phi=\{(0,0),(1,0),(0,1)\}$, the Euclidean origin is in NNR with both $(1,0)$ and $(0,1)$. 
\begin{defi}
\label{Single}
Two different atoms $x,y$ are in Mutually Nearest Neighbor Relation (MNNR) if and only if (iff) $x \stackrel{\phi}{\rightarrow} y$ and $y \stackrel{\phi}{\rightarrow} x$, and we denote it by $x \stackrel{\phi}{\leftrightarrow} y$. In telecommunication terms, we say that the two BSs $x$ and $y$ are in cooperation. 
\end{defi} 

\begin{defi}
\label{Pair}
An atom $x\in \phi$ is called single iff it is not in MNNR (does not cooperate) with any other atom in $\phi$. That is, if for every $y\in \phi $ such that $x \stackrel{\phi}{\rightarrow} y$, then $y \stackrel{\phi}{\not \rightarrow} x$.
\end{defi}

\begin{figure}[htbp] 
\centering
\subfigure[]{\includegraphics[trim = 0mm 0mm 0mm 0mm, clip,width=0.3\textwidth]{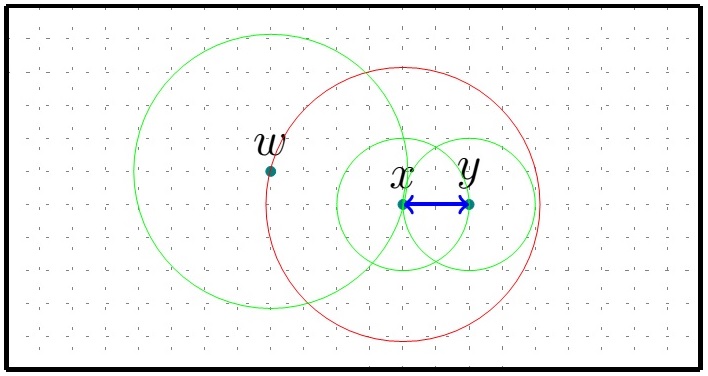}}
\subfigure[]{\includegraphics[trim = 20mm 100mm 20mm 100mm, clip,width=0.45\textwidth]{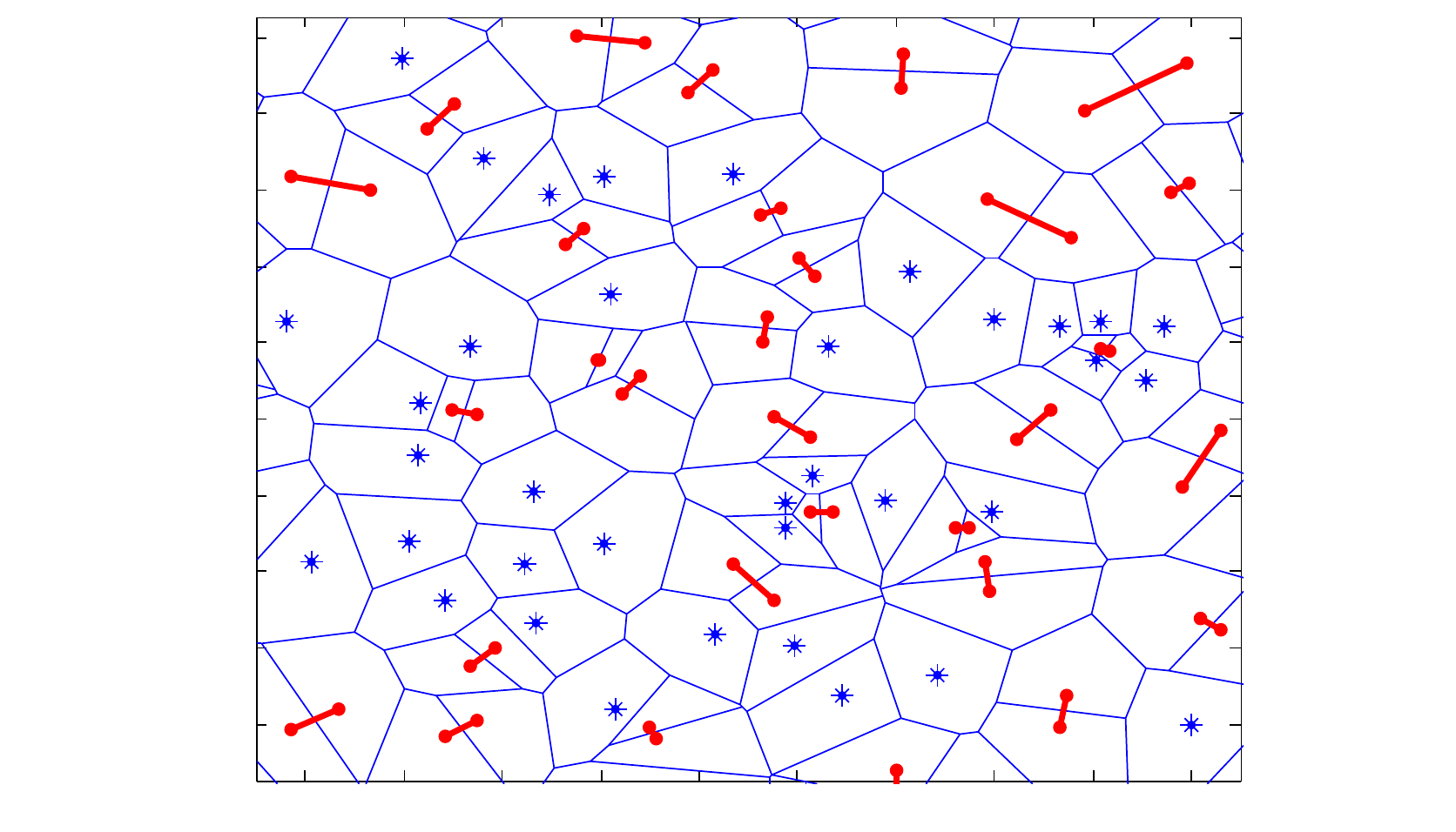}} 
\caption{(a)\ The atoms $x$ and $y$ are mutually nearest neighbors, so, they work in pair. The atom $x$ is the nearest neighbor of $w$, but $w$ is not the closest atom to $x$, thus $w$ is single. (b) A Poisson realisation with its corresponding Voronoi diagram. The asterisks are the single BSs, the connected dots are the cooperating pairs. }   \label{Groups}
\end{figure}
For $x,y\in \phi$ fixed, denote the area
\begin{equation*}
C(x,y):=B(x,\|x-y\|)\cup B(y,\|x-y\|).
\end{equation*}
In geometric terms, the relation $x \stackrel{\phi}{\rightarrow} y$ holds iff the disc $B(x,\|x-y\|)$ is empty of atoms from $\phi$. Consequently, the relation $x \stackrel{\phi}{\leftrightarrow} y$ holds iff, there are no atoms from $\phi$ inside $C(x,y)$. The Euclidean surface of $C(x,y)$ is $\pi\|x-y\|^2(2-\gamma)$, where $\gamma:=\frac{2}{3}-\frac{\sqrt{3}}{2\pi}$ is a constant number equal to the surface, divided by $\pi$, of the intersection of two discs with unit radius and centres lying on the circumference of each other. An illustration of the above explanations is given in Figure \ref{Groups}. 

When $\Phi=\{ \phi \}$ is a PPP, we have an expression of its empty space function. With this in mind, along with the above argument, it is possible to give a closed form to the probability of two given atoms being in pair. 

\begin{lem} \label{Lemma1}
Given a PPP $\Phi$, with density $\lambda>0$, for two different and fixed atoms $x,y\in \mathbb{R}^2$, 
\begin{equation*}
\mathbb{P}(x \stackrel{\Phi}{\leftrightarrow} y)=e^{-\lambda \pi\|x-y\|^2(2-\gamma)}.
\end{equation*}
\end{lem}
For a stationary point process $\Phi$, define two new point processes $\Phi^{(1)}$ and $\Phi^{(2)}$, that result from the dependent thinning defined above:
\begin{equation*}
\begin{split}
\Phi^{(1)}&:=\{ x\in \Phi \ \& \ x \mbox{ is single} \}, \\
\Phi^{(2)}&:=\{ x\in \Phi \ \& \ x \mbox{ cooperates with another element of } \Phi \}. 
\end{split}
\end{equation*}
Both processes are \textit{stationary}. This is due to the stationarity of $\Phi$ and because, by definition, they depend only on the distance between the elements of $\Phi$. From the previous Lemma, Slivniak's Theorem and Campbell-Little-Mecke formula \cite{BacBlaVol1} we have the following result.

\begin{theo}
\label{Percentage}
Given a PPP $\Phi$, with density $\lambda>0$, for every fixed atom $x\in \mathbb{R}^2$, there exists a constant $\delta>0$, independent of $\lambda$ and $x$, such that 
\begin{equation*}
\begin{split}
\mathbb{P} \left( x \in\Phi^{(2)} \right)=\delta, \ \ \mathbb{P} \left( x\in\Phi^{(1)} \right)=1-\delta.
\end{split}
\end{equation*} 
Specifically, $\delta=\frac{1}{2-\gamma} \approx 0.6215$.
\end{theo}

\begin{proof}
By definition of $\Phi^{(2)}$,
\begin{equation*}
\begin{split}
\mathbb{P}\left( x \in \Phi^{(2)}\right) & = \mathbb{P}\left( x \stackrel{\Phi}{\leftrightarrow} y, \mbox{ for some } y\in \Phi \backslash \{x\} \right) \\
& = \mathbb{E}\left( \textbf{1}_{\left\{ x \stackrel{\Phi}{\leftrightarrow} y, \mbox{ for some } y\in \Phi \backslash \{x\} \right\}} \right) \\
& \stackrel{(a)}{=} \mathbb{E} \left( \sum_{y\in \Phi } \textbf{1}_{ \left\{ x \stackrel{\Phi}{\leftrightarrow} y \right\} } \right),
\end{split}
\end{equation*}
where $(a)$ holds because for PPPs the nearest neighbor is a.s. unique. Using Campbell-Little-Mecke formula \cite{BacBlaVol1} for a PPP, 
\begin{equation*}
\begin{split}
\mathbb{E} \Bigg( \sum_{y\in \Phi } \textbf{1}_{ \left\{ x \stackrel{\Phi}{\leftrightarrow} y \right\} } \Bigg) & = \int_{\mathbb{R}^2}\mathbb{P} \left( x \stackrel{\Phi}{\leftrightarrow} y \right) \lambda dy \\
& \stackrel{(b)}{=} \lambda \int_{\mathbb{R}^2} e^{-\lambda \pi \|x-y \|^2}(2-\gamma) dy \\
& = \frac{1}{2-\gamma},
\end{split}
\end{equation*}
where $(b)$ follows from Lemma \ref{Lemma1}.
\end{proof}

\textbf{Remark:} In Lemma \ref{Lemma1} and Theorem \ref{Percentage} we actually make use of the Palm measures of the process, but avoid its notation for ease of presentation, without substantial difference.   

The constant $\delta$ is crucial within this work. The above Theorem states that, given the position of a BS (in a PPP), its probability of being in a cooperation pair is $\delta\approx 0.6215$, otherwise, its probability of being single is $1-\delta\approx 0.3785$, irrespective of the value of the density $\lambda>0$. Since we are fixing the atom location, this result should be interpreted from a local point of view. Nevertheless, in Section \ref{SecIII} we will prove that, for a given density of the PPP $\lambda>0$, the intensities of $\Phi^{(1)}$ and $\Phi^{(2)}$ are actually $(1-\delta)\lambda$ and $\delta \lambda$, respectively. The former can be interpreted from a global point of view: over any planar area in $\mathbb{R}^2$, in average, $37.85\%$ of atoms are singles and $62.15\%$ belong to a cooperative pair. 

When $\Phi$ is a PPP, it is natural to wonder if $\Phi^{(1)}$ and $\Phi^{(2)}$ are also PPPs. As a matter of fact, they are not (we could have expected this, since they were defined by a strongly dependent thinning). Suppose that $\Phi^{(2)}$ is actually a PPP. As shown in Theorem \ref{Percentage}, for every atom in  $\Phi^{(2)}$, there is a positive probability of this point not being in MNNR with another point of $\Phi^{(2)}$. However, by definition, all the elements of $\Phi^{(2)}$ are in MNNR with another element of $\Phi^{(2)}$, which is a contradiction. We conclude that the process $\Phi^{(2)}$ is not a PPP. For $\Phi^{(1)}$ the argumentation is not as simple. We can show using the Kolmogorov-Smirnov test \cite{Sheskin07} that the number of $\Phi^{(1)}$ atoms within a finite window is not Poisson distributed. Moreover, Monte Carlo simulations estimate that the average proportion of single atoms from $\Phi^{(1)}$ is far from the $37.85\%$.

We can show that the percentages in Theorem \ref{Percentage} are not valid just for PPPs. Take the hexagonal grid model. This is commonly used by industry related research teams to model the BS positions, and then evaluate a system deployment and performance via Monte Carlo methods. The hexagonal grid's centers should represent the BS locations. This is an ideal scenario (the BSs are never that regular). We introduce another point process, based on the hexagonal grid, that actually allows for randomnes of the BS positions. Starting from the grid placement, let the position of each BS be randomly perturbed, independently of the others. For example, consider as BS location the point whose polar coordinates around each hexagon's center follow two uniform random variables (r.v.s), one angular over $[0,2\pi)$ and the radial one over $[0,Q]$ (see Figure \ref{HexGridPerc}). Figure \ref{HexGridPerc} shows how the average percentage of singles and pairs for the hexagonal grid model changes when varying the parameter $Q>0$. Remark that these numbers are very close to the respective average percentages we found when $\Phi$ is a PPP.    

\begin{figure}[htbp] 
\centering
\subfigure[]{\includegraphics[trim = 35mm 70mm 20mm 70mm, clip,width=0.23\textwidth]{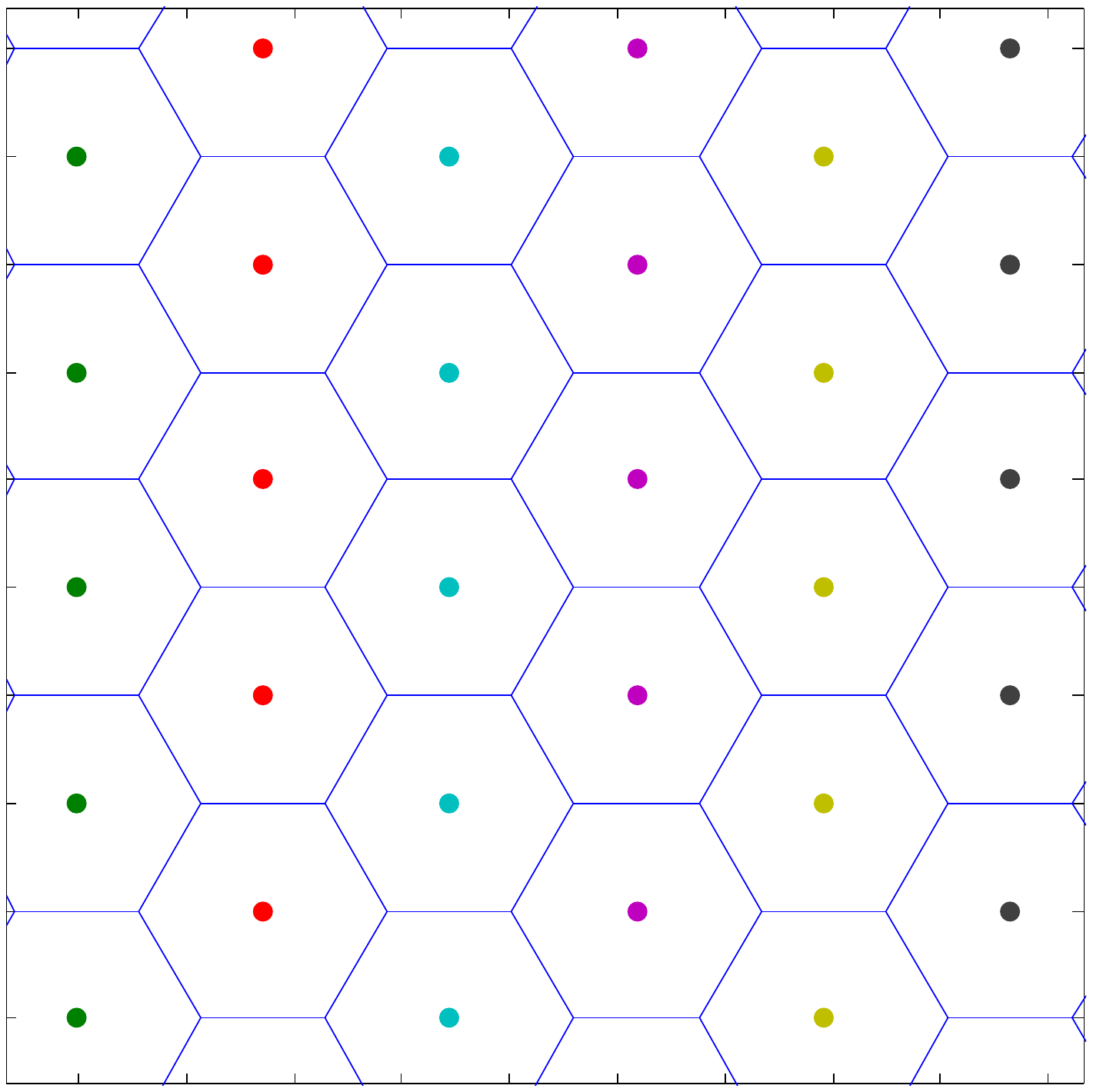}}
\subfigure[]{\includegraphics[trim = 35mm 70mm 20mm 70mm, clip,width=0.23\textwidth]{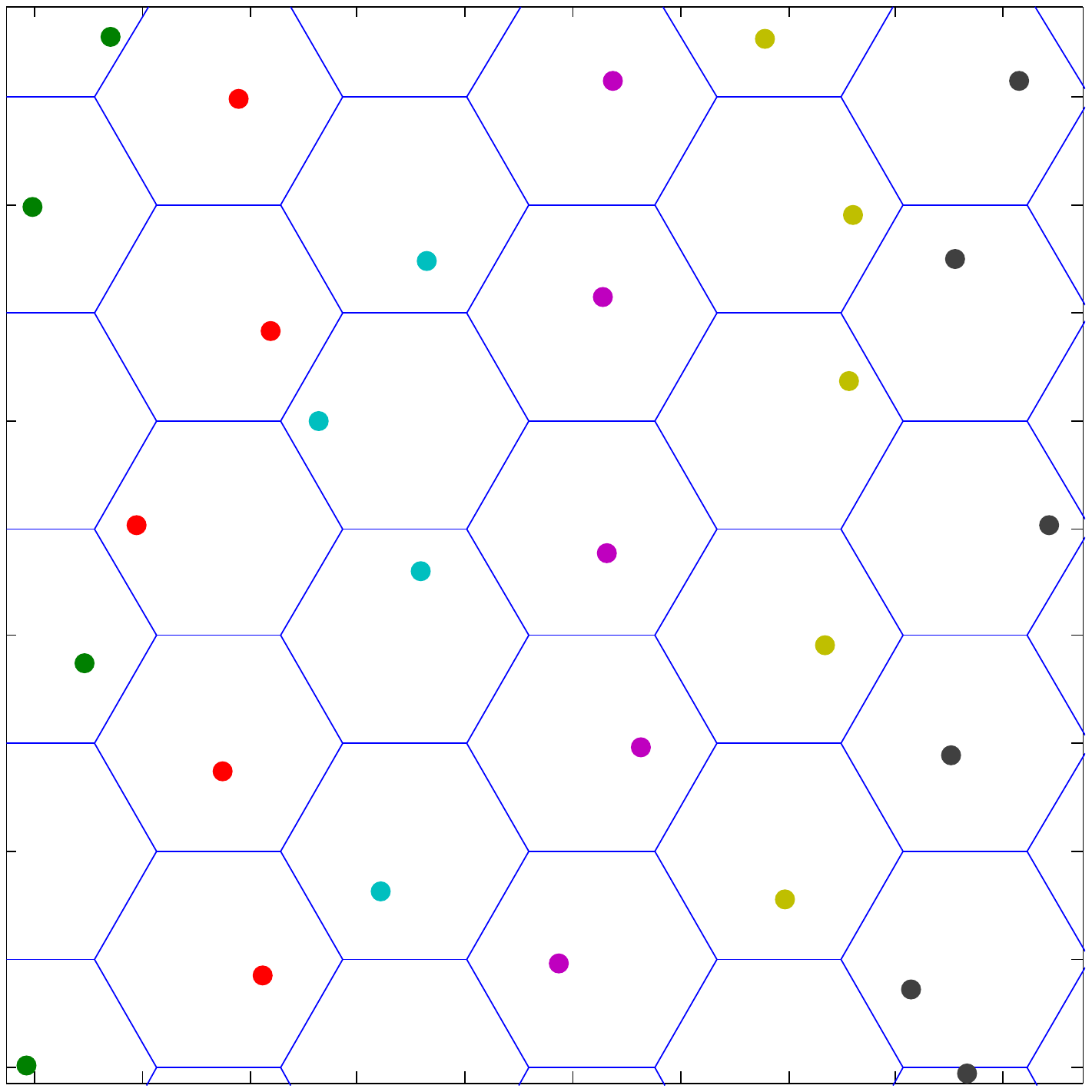}} 
\subfigure[]{\includegraphics[trim = 35mm 100mm 40mm 100mm, clip,width=0.23\textwidth]{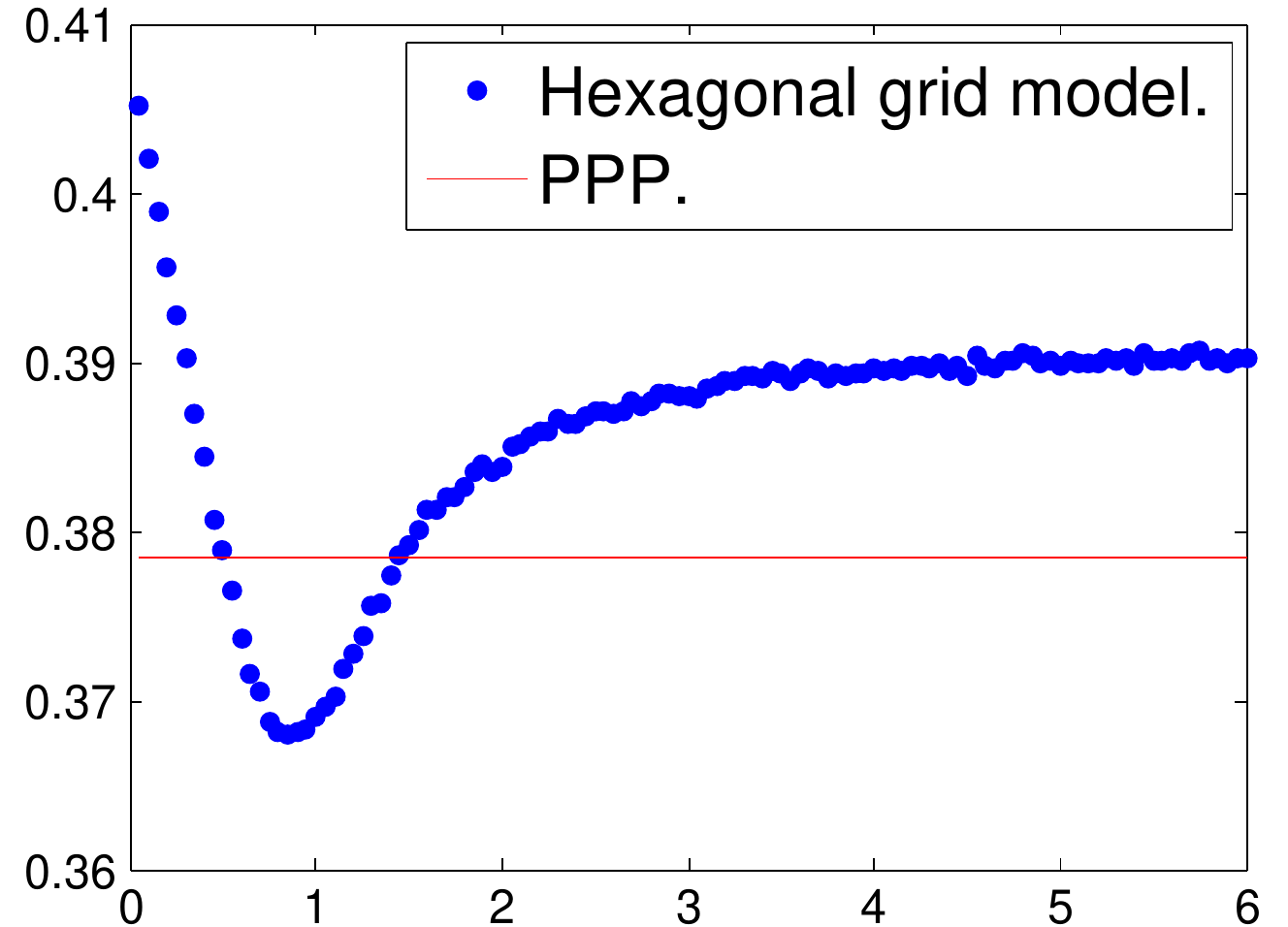}}
\subfigure[]{\includegraphics[trim = 35mm 100mm 40mm 100mm, clip,width=0.23\textwidth]{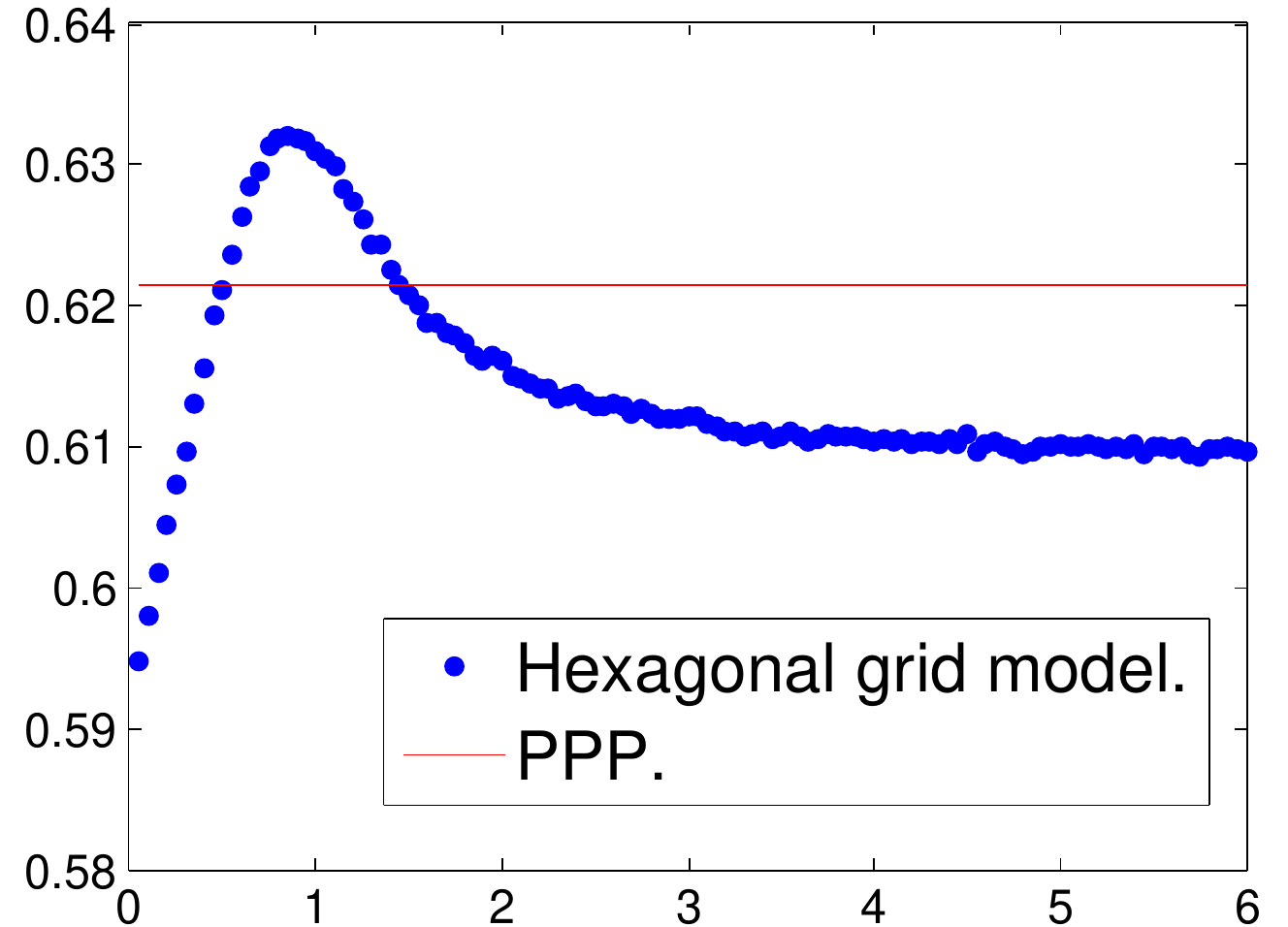}} 
\caption{(a) The hexagonal grid model, without perturbation. (b) The hexagonal grid model, with the centers being perturbed via a random experiment. (c) The average percentage of single atoms for the hexagonal grid model. (d) The average percentage of atoms in cooperative pairs.} 
\label{HexGridPerc}
\end{figure}

\subsection{Palm Probabilities}
\label{SecIIC}

We can interpret the Palm probability of a stationary point process as the conditional probability, given that the process has a point inside an infinitesimal neighborhood around a fixed atom \cite{BacBlaVol1}. Denote by $\mathbb{P}^{0}$, $\mathbb{P}^{(1),0}$, and $\mathbb{P}^{(2),0}$ the Palm probabilities of the stationary point processes $\Phi$, $\Phi^{(1)}$, and $\Phi^{(2)}$, respectively. Let
\begin{equation*}
\begin{split}
A_0 := \{\Phi \in \mathcal{A}_0 \},\ \ \ B_0 := \{\Phi \in \mathcal{B}_0  \},
\end{split}
\end{equation*}
be two events, where 
\begin{equation*}
\begin{split}
\mathcal{A}_0 & := \{\phi \ | \ 0\in \phi \mbox{ and $0$ is single  } \},\\
\mathcal{B}_0 & := \{\phi \ | \ 0\in \phi \mbox{ and $0$ cooperates with another atom of $\phi$ } \}.
\end{split}
\end{equation*}
We have the following result \cite[pp. 35, Ex. 142]{EleQueTh2003}.

\begin{theo}
\label{Palm}
Let $\Phi$ be a stationary point process such that $\mathbb{P}^{0}(A_0)>0$ and $\mathbb{P}^{0}(B_0)>0$. Therefore, for every $C \in \Omega$,
\begin{equation*}
\begin{split}
\mathbb{P}^{(1),0}(C) =\mathbb{P}^{0}(C|A_0), \ \  \mathbb{P}^{(2),0}(C) =\mathbb{P}^{0}(C|B_0).
\end{split}
\end{equation*}
\end{theo}
When $\Phi$ is a PPP, $\mathbb{P}^{0}(A_0)=1-\delta>0$ and $\mathbb{P}^{0}(B_0)=\delta>0$ (Theorem \ref{Percentage}). Then, for every $C\in \mathcal{F}$,
\begin{equation}
\label{PalmPoisson}
\begin{split}
\mathbb{P}^{(1),0}(C) = \frac{\mathbb{P}^{0}(C,A_0)}{1-\delta}, \ \ \mathbb{P}^{(2),0}(C) = \frac{\mathbb{P}^{0}(C,B_0)}{\delta}.
\end{split}
\end{equation}

\subsection{The NN function of $\Phi^{(2)}$}

The \textit{Nearest Neighbor function (NN)}, commonly denoted by $G$, is the cumulative distribution function (CDF) of the distance from a typical atom of the process to its nearest neighboring point \cite{BaddNotes07}. 
Denote by $G^{(2)}(r)$ the NN function of $\Phi^{(2)}$, then, 
\begin{eqnarray}
G^{(2)}(r) = \mathbb{P}^{(2),0}(d(0,\Phi^{(2)}\setminus \{0\})\leq r)\nonumber,
\end{eqnarray}   
for every $r>0$. Applying equation \eqref{PalmPoisson} to the above expression, we have the following result.

\begin{theo}
\label{TheoNN2}
For a PPP $\Phi$, the NN function of $\Phi^{(2)}$ is
\begin{eqnarray}
\label{NNPhi2}
G^{(2)}(r)  = & 1-e^{-\lambda\pi r^2 (2-\gamma)},
\end{eqnarray}   
where $\gamma$ is the same constant as in Lemma \ref{Lemma1}.
\end{theo}

\begin{proof}
Remark that under $\mathbb{P}^{0}$
\begin{equation}
\label{disjointU}
A_0 = \bigcup_{y\in \Phi\backslash \{0\}}\{ 0 \stackrel{\Phi}{\leftrightarrow} y\}
\end{equation}
where $A_0$ is the event from Theorem \ref{Palm}, and this union is mutually disjoint. For $r>0$ fixed,
\begin{equation*}
\begin{split}
G^{(2)}(r) & = \mathbb{P}^{(2),0}\left( d(0,\Phi^{(2)}\backslash \{0\})\leq r \right) \\
& \stackrel{(a)}{=} \mathbb{P}^{0}\left( d(0,\Phi^{(2)}\backslash \{0\})\leq r , A_0 \right) \frac{1}{\delta} \\
& \stackrel{(b)}{=} \mathbb{E}^{0}\left( \sum_{y\in \Phi \backslash \{0\}} \textbf{1}_{\{ d(0,\Phi^{(2)}\backslash \{0\})\leq r , 0 \stackrel{\Phi}{\leftrightarrow} y \} }  \right) \frac{1}{\delta}, \\
& \stackrel{(c)}{=} \mathbb{E}\left( \sum_{y\in \Phi } \textbf{1}_{\{ d(0,\Phi^{(2)}\backslash \{0\})\leq r , 0 \stackrel{\Phi}{\leftrightarrow} y \} }  \right) \frac{1}{\delta},
\end{split}
\end{equation*}
where $(a)$ follows from equation \eqref{PalmPoisson}, $(b)$ after equation \eqref{disjointU}, and $(c)$ from Slivkyak-Mecke's Theorem. Observe that, if there is some $y\in \Phi $ being the mutually nearest neighbor of the atom $0$, that is $0 \stackrel{\Phi}{\leftrightarrow} y$, then,  
\begin{equation*}
d(0,\Phi^{(2)}\backslash \{0\}) = d(0,\Phi \backslash \{0\})=\|y\| \ \ a.s.
\end{equation*}
Using this, Campbell-Little-Mecke formula and Lemma \ref{Lemma1}, 
\begin{equation*}
\begin{split}
\mathbb{E}\Bigg( \sum_{y\in \Phi } & \textbf{1}_{\{ d(0,\Phi^{(2)}\backslash \{0\})\leq r , 0 \leftrightarrow y \} }  \Bigg) \frac{1}{\delta} \\
& = \mathbb{E}\Bigg( \sum_{y\in \Phi } \textbf{1}_{\{ \|y\|\leq r , 0 \stackrel{\Phi}{\leftrightarrow} y \} }  \Bigg) \frac{1}{\delta} \\
& = \mathbb{E}\Bigg( \sum_{y\in \Phi } \textbf{1}_{\{ \|y\|\leq r\} }\textbf{1}_{\{ 0 \stackrel{\Phi}{\leftrightarrow} y \} }  \Bigg) \frac{1}{\delta} \\
& = \int_{\mathbb{R}^2}  \mathbb{E}\Big( \textbf{1}_{\{ \|y\|\leq r\} }\textbf{1}_{\{ 0 \stackrel{\Phi}{\leftrightarrow} y \} }  \Big) \lambda dy \frac{1}{\delta} \\
& = \int_{\mathbb{R}^2} \mathbf{1}_{\{ \|y\|\leq r\}}  \mathbb{E}\Big( \textbf{1}_{\{ 0 \stackrel{\Phi}{\leftrightarrow} y \} }  \Big) \lambda dy \frac{1}{\delta} \\
& = \frac{\lambda}{\delta} \int_{\{ \|y\|\leq r\}}  \mathbb{P}\big( 0 \stackrel{\Phi}{\leftrightarrow} y \big)  dy  \\
& = \frac{\lambda}{\delta} \int_{\{ \|y\|\leq r\}}  e^{-\lambda \pi \|y\|^2(2-\gamma)}  dy  \\
& \stackrel{(d)}{=} \frac{\lambda 2\pi}{\delta} \int^r_0  e^{-\lambda \pi s^2(2-\gamma)}s ds   \\
& = 1-e^{-\lambda \pi r^2(2-\gamma)},
\end{split}
\end{equation*}
where $(d)$ follows from the change of variable to polar coordinates.
\end{proof}

The last Theorem simply states that, in the PPP case, the distance between cooperative atoms is Rayleigh distributed, with scale parameter $\alpha:=(2\lambda\pi(2-\gamma))^{-1/2}$. 

\subsection{Size of the Voronoi Cells}

It follows naturally to investigate the size of Voronoi cells associated with single atoms or pairs. A Voronoi cell of an atom $x\in \phi$ is defined to be the geometric locus of all planar points $z\in \mathbb{R}^2$ closer to this atom than to any other atom of $\phi$ \cite{CompGeomBook}. In a wireless network the Voronoi cell is important when answering the question 'which user should be associated with which station?'. 

In a stationary framework, we examine the network performance at the Cartesian origin, the \textit{typical user approach}. Let $\{0\curvearrowright \Phi^{(1)}\}$ (resp. $\{0\curvearrowright \Phi^{(2)}\}$) denote the event that the typical user belongs to the Voronoi cell of some atom of $\Phi^{(1)}$ (resp. $\Phi^{(2)}$). For the PPP case we have the following result.

\begin{prop}
\label{VoronoiPerc}
Suppose that $\Phi$ is a PPP, with density $\lambda>0$. There exists a measurable function $F:[0,\infty)\times [0,\infty) \times [0,2\pi) \times [0,2\pi) \longrightarrow [0,\infty)$, such that 
\begin{equation*}
\begin{split}
\mathbb{P}& (0\curvearrowright \Phi^{(2)}) \\
& = \lambda^2 \int^\infty_0 \int^\infty_0 \int^{2\pi}_0 \int^{2\pi}_0 s r e^{-\lambda F(r,s,\theta,\varphi) -\lambda \pi r^2} d\varphi d\theta  dr ds 
\end{split}
\end{equation*}
\end{prop}

\begin{proof}
See Appendix A, available in the supplemental material.
\end{proof}
Since
\begin{equation}
\mathbb{P}(0\curvearrowright \Phi^{(1)})=1-\mathbb{P}(0\curvearrowright \Phi^{(2)}),
\end{equation} 
we have also an analytic representation for $\mathbb{P}(0\curvearrowright \Phi^{(1)})$. The function $F(r,s,\theta,\varphi)$ is not explicitly given, being the Euclidean surface of three overlapping discs. This is an example of the complications that arise from the MNNR, due to numerical issues related to integration over multiple overlapping circles. Such complications led us to the approximate model in Section \ref{SecV}.

\begin{NR}
Given a PPP $\Phi$, the average surface proportion of Voronoi cells associated with single atoms, and that associated with pairs of atoms, is independent of the parameter $\lambda$. By Monte Carlo simulations, we find that
\begin{equation*}
\begin{split}
\mathbb{P}(0\curvearrowright \Phi^{(1)}) \approx 0.4602, \ \ \mathbb{P}(0\curvearrowright \Phi^{(2)}) \approx 0.5398.
\end{split}
\end{equation*}
\end{NR} 
Interestingly, although the ratio of single atoms to pairs is $0.3785/0.6215 \approx 0.6090$, the ratio of the associated Voronoi surface is $0.4602/0.5398 \approx 0.8525$, implying that the typical Voronoi cell of a single atom is larger than that of an atom from a pair, as Figure \ref{Groups} shows. The last remark gives a first intuition that there is attraction between the cooperating atoms in pair and repulsion among the singles. 

\subsection{Further Results}

The \textit{empty space function (ES)}, commonly denoted by \textit{F}, is the CDF of the distance from the typical user to the nearest atom of the point process considered \cite{BaddNotes07}. The two functions NN and ES can be combined into a single expression known as the \textit{$J$ function}. The latter is a tool introduced by van Lieshout and Baddeley \cite{BaddNotes07} to measure repulsion and/or attraction between the atoms of a point process. It is defined as
\begin{equation}
J(r)=\frac{1-G(r)}{1-F(r)},
\end{equation} 
for every $r>0$. In the case of the PPP, $G\left(r\right)\equiv F\left(r\right)$ and $J\left(r\right)=1$, as a consequence of the fact that the reduced Campbell measure is identical to the original measure. Hence the $J$ function quantifies the differences of any process with the PPP. When $J(r)>1$, this is an indicator of repulsion between atoms, whereas $J(r)<1$ indicates attraction. 
We use Monte Carlo simulations to plot the $J$ function of both processes (see Figure \ref{JFunction}). From the figures we conclude that $\Phi^{(1)}$ \textit{exhibits repulsion for every $r\geq 0$, and $\Phi^{(2)}$ attraction everywhere}. However, note that the attraction in the case $\Phi^{(2)}$ is due to the way the pairs were formed. If we consider a new process having as elements the middle points between each one of the cooperating pairs, this process exhibits repulsion everywhere.

\begin{figure}[htbp] 
\centering
\subfigure[]{\includegraphics[trim = 45mm 100mm 40mm 100mm, clip,width=0.23\textwidth]{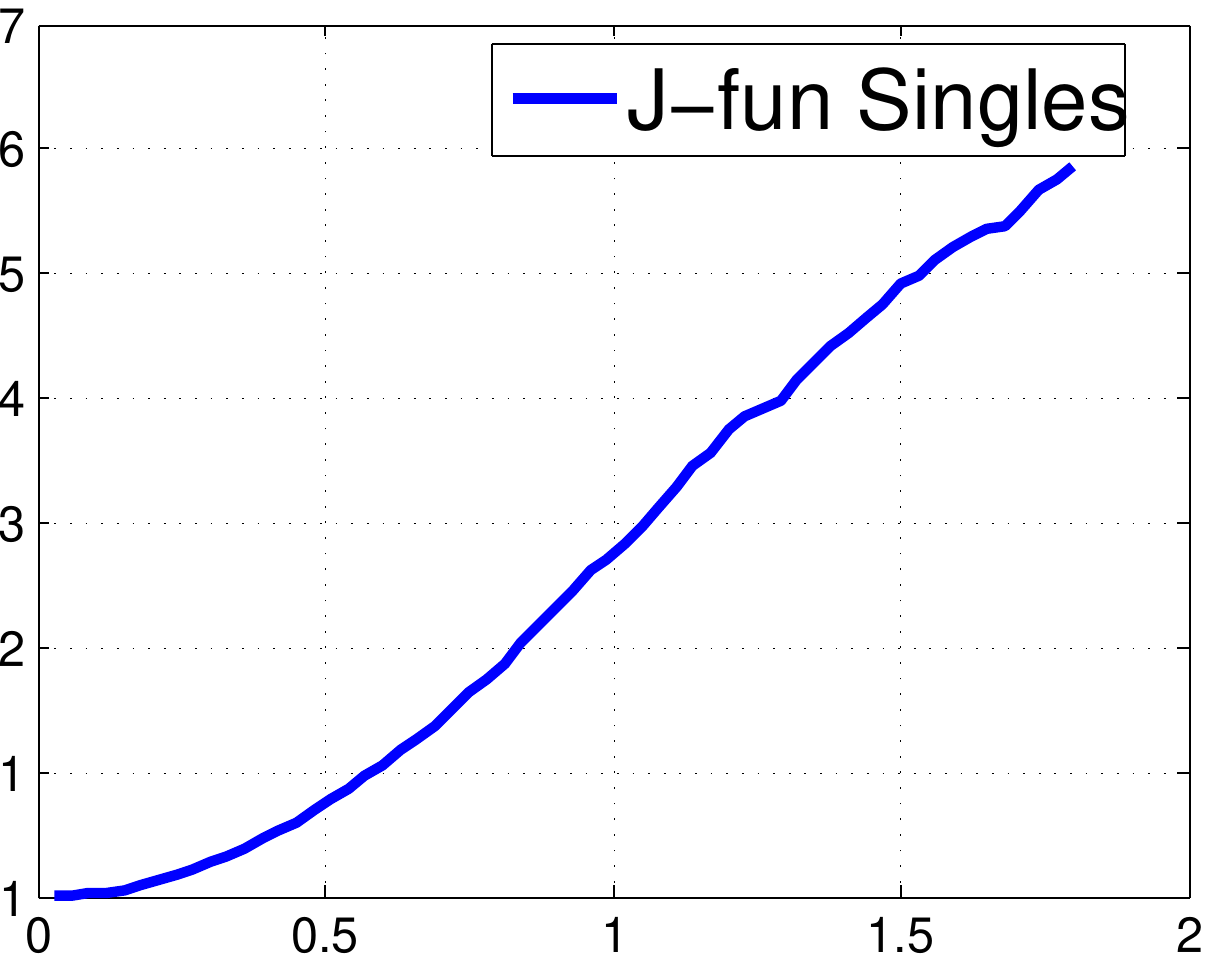}}
\subfigure[]{\includegraphics[trim = 45mm 100mm 40mm 100mm, clip,width=0.23\textwidth]{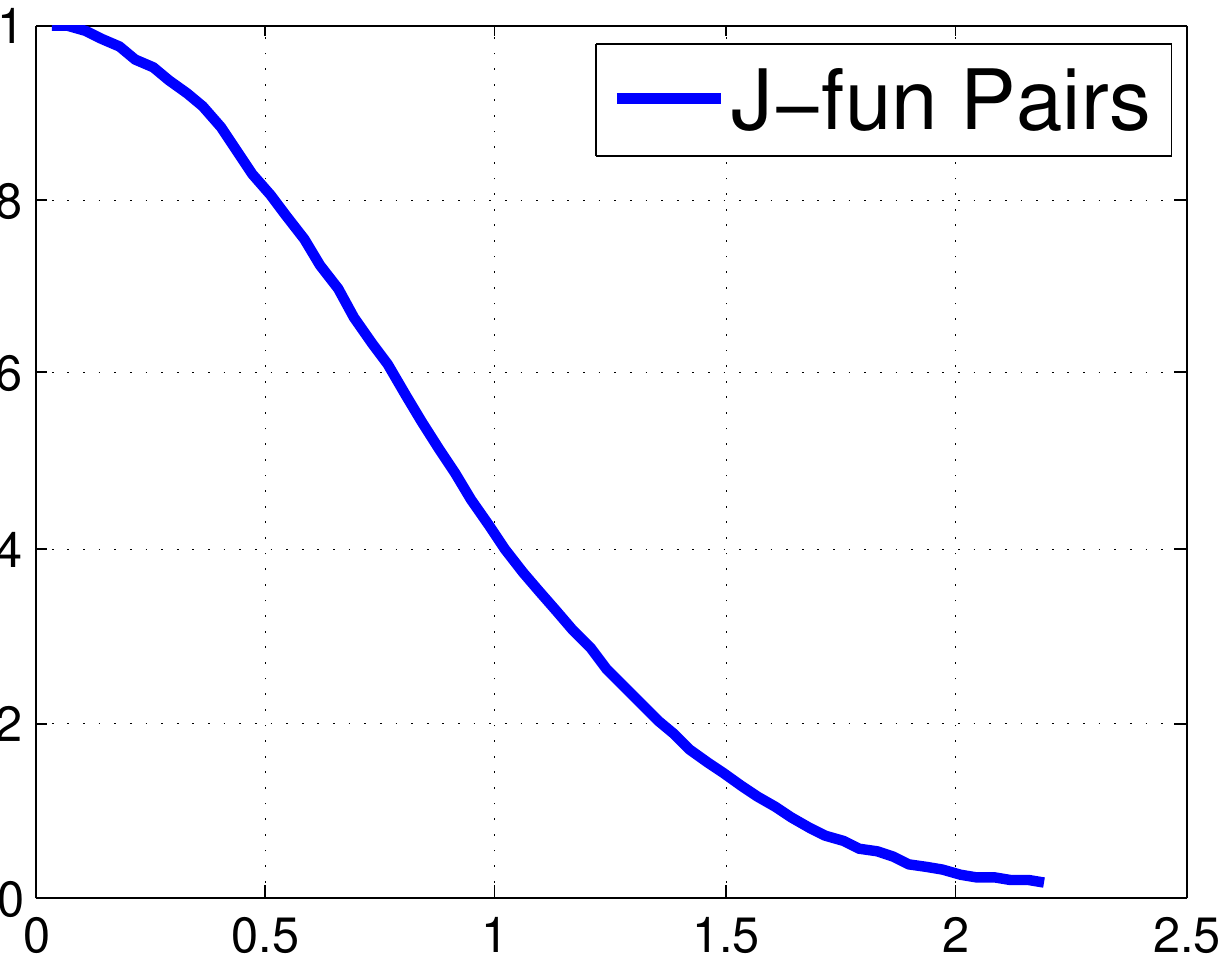}} 
\label{JFunction}
\caption{(a) The J function of the processes $\Phi^{(1)}$. (b) The J function of the processes $\Phi^{(2)}$.} 
\label{JFunction}
\end{figure}

\section{Received Signals}
\label{SecIII}
The analysis on this work can be applied to any type of antennas, including directional ones. For simplicity of presentation we will treat here the omnidirectional case, where the emitted signal depends only on the distance of the BS from the typical user. The case of directional BSs requires extra integration with respect to angles, which unnecessarily complicates the analysis, without substantial difference.

In what follows, we will introduce explicit examples. For these, let us consider an i.i.d. family $(h_r)_{r>0}$ of positive exponential variables, with parameter $1$, also independent of the BS positions. Given $p>0$, the couple $(h_r,p)$ represents the random propagation effects and the power signal emitted to the typical user from a BS whose distance from the origin is $r>0$. Let us also choose the \textit{path-loss} function as $l(r):=\frac{1}{r^\beta}$,  with \textit{path loss exponent} $\beta>2$. 

\subsection{Single atoms}
Consider $f:\mathbb{R}^2 \rightarrow [0,\infty)$ and $\tilde{f}:[0,\infty) \longrightarrow \mathbb{R}^+$ two generic random fields. The quantity $f(x)$ (and $\tilde{f}(r)$) represents the received signal at the typical user, when transmitted by a single BS, whose position is $x$ (or its distance from the origin is $r>0$). For a single BS we could, for example, consider 
\begin{equation}\label{signalsin}
\tilde{f}(r)=p \frac{h_r} {r^\beta},
\end{equation}  
which follows an exponential distribution, with parameter $\frac{r^\beta}{p}$. 

\subsection{Pair cooperation}
Consider $g:\mathbb{R}^2 \times \mathbb{R}^2 \longrightarrow \mathbb{R}^+$ and $\tilde{g}:[0,\infty)\times [0,\infty) \longrightarrow \mathbb{R}^+$ two generic random fields, both independent of the BS positions. The quantity $g(x,y)$ (and $\tilde{g}(r,z)$) represents the received signal at the typical user, when transmitted by a pair of BSs whose positions are $x$ and $y$ (or their distance from the origin are $r>0$ and $z>0$), respectively. The received signal can take the following example expressions, which refer to different types of cooperation or coordination,
\begin{equation}\label{coopFunction}
\tilde{g}(r,z)=
\begin{cases}
p \frac{h_r}{r^\beta} + p \frac{h_z}{z^\beta}  \ , \ \ \ \ \ \ \ \ \ \ \ \ \ \ \ \ \ \ \ \ \mbox{\textbf{[NSC]}} \\
\mathbf{1}_{on_r}p \frac{h_r}{r^\beta} + (1-\mathbf{1}_{on_r})p \frac{h_z}{z^\beta} \ , \ \ \mbox{\textbf{[OFF]}} \\
\max \left\{p \frac{h_r}{r^\beta},p \frac{h_z}{z^\beta}\right\} \ , \ \ \ \ \ \ \ \ \ \ \ \ \ \mbox{\textbf{[MAX]}} \\
\left| \sqrt{p \frac{h_r}{r^\beta}}e^{i\theta_r}+\sqrt{p \frac{h_z}{z^\beta}}e^{i\theta_z}\right|^2 \ \ \ \ \mbox{\textbf{[PH]}}
\end{cases}.
\end{equation}
In the above, $(\mathbf{1}_{on_r})_{r>0}$ and $(\theta_r)_{r>0}$ are two different families of indexed identically distributed r.v.s, independent of the other random objects. They follow a Bernoulli distribution, with parameter $q\in (0,1)$ ($\overline{q}:=1-q$), and a general distribution over $[0,2\pi)$, respectively. \textbf{[NSC]} refers to non-coherent joint transmission, as in  \cite{NigCoordMul2014,BlasStuSINTFact2015,TanTracMod2014,GuoSPGP2014}, where each of the two BSs transmits an orthogonal signal, and the two are added at the receiver side. \textbf{[OFF]} refers to the case where one of the two BSs is active and the other inactive, according to an independent Bernoulli experiment, independent of the BS positions. \textbf{[MAX]} refers to the case where the BS with the strongest signal is actively serving a user, while the other is off. The \textbf{[OFF]} and \textbf{[MAX]} cases are relevant to energy saving operation. In the \textbf{[PH]} case, two complex signals are combined in phase (see \cite{BacAStoGeo2015,NigCoordMul2014}), in particular, when $cos(\theta_r-\theta_z)=1$, the two signals are in the same direction, and they add up coherently at the receiver (user side), giving the maximum cooperating signal. 

The above expressions in \eqref{coopFunction} are merely examples of the cooperation signals. A more general family can be proposed with specific properties. Consider $c_i:[0,\infty)\times[0,\infty) \longrightarrow \mathbb{R}$, and $d_i:[0,\infty)\times[0,\infty) \longrightarrow \mathbb{R}^+$, for $1\leq i \leq n$, some deterministic and measurable functions, and suppose that  
\begin{equation}\label{TAIL}
\mathbb{P}(\tilde{g}(r,z)>T)=\sum^n_{i=1}c_i(r,z)e^{-d_i(r,z)T}.
\end{equation}
When analysing performance related to coverage probability, the tail probability distribution functions (CCDF) for the signals that can be written as \eqref{TAIL} lead easier to numerically tractable formulas. However, the function defined in \eqref{TAIL} is not necessarily a CCDF. For this to hold, some extra conditions must be imposed to the functions 
$c_i(r,z)$ and $d_i(r,z)$.
Interestingly, the CCDF of $\tilde{g}(r,z)$ in the \textbf{[NSC]}, \textbf{[OFF]}, and \textbf{[MAX]} cases fulfils equation \eqref{TAIL} (see Table \ref{table}). Furthermore, there exist important families of r.v.s whose CCDF actually has the form described in equation \eqref{TAIL}: the \textit{hypo-exponential distribution}, the \textit{hyper-exponential distribution}, the maximum over a finite number of exponential r.v.s, among others. 
\begin{table} 
	\centering
	\caption{Expressions for the CCDF and the LT}\label{table}
    \begin{tabular}{ | l | p{4cm} | p{2.5cm} | } 
    \hline
     \ \ & \centering $\mathbb{P}(g(r,z)>T)$  & $\mathbb{E}[e^{-s g(r,z)}]$  \\ \hline
    \textbf{[NSC]} & $\frac{z^\beta}{p(r^\beta-z^\beta)} \Big(e^{-\frac{r^\beta}{p}T}-e^{-\frac{z^\beta}{p}T} \Big)$ & $\frac{r^\beta}{sp+r^\beta}\frac{z^\beta}{sp+z^\beta}$
    \\ \hline
    \textbf{[OFF]} & $qe^{-\frac{r^\beta}{p}T}+\overline{q}e^{-\frac{z^\beta}{p}T}$ & $q\frac{r^\beta}{sp+r^\beta}+\overline{q}\frac{z^\beta}{sp+z^\beta}$ 
    \\ \hline
    \textbf{[MAX]} & $e^{-\frac{r^\beta}{p}T}+e^{-\frac{z^\beta}{p}T}-e^{- \left( \frac{ r^\beta}{p} + \frac{z^\beta}{p}\right)T}$ & $\frac{r^\beta}{sp+r^\beta}+\frac{z^\beta}{sp+z^\beta}-\frac{r^\beta+z^\beta}{sp+r^\beta+z^\beta}$
    \\    \hline
    \end{tabular}
\end{table}

\section{Interference Analysis in the MNNR model}
\label{SecIV}

The purpose of the analysis up to this point was to develop the basic tools, within a communication context, that will allow us to derive results related to cooperation. As shown in the previous Section, the cooperating BS pairs will have a different influence on the interference seen by a user in the network, than those operating individually. The current Section will focus on the \textit{interference field} generated by $\Phi^{(1)}$ and $\Phi^{(2)}$. As shown in Section \ref{SecII}, even when $\Phi$ is a PPP, the two processes behave differently than a PPP. We thus have to resort to direct techniques from the theory of Stochastic Geometry and point processes. 

If we denote by $\mathcal{I}^{(1)}$ and $\mathcal{I}^{(2)}$, the interference field generated by the elements of $\Phi^{(1)}$ and $\Phi^{(2)}$, then, 
\begin{eqnarray}
\label{I1}
\mathcal{I}^{(1)} & = & \sum_{x\in \Phi^{(1)}} f(x),\\
\label{I2}
\mathcal{I}^{(2)} & = & \frac{1}{2}\sum_{x\in \Phi^{(2)}} \sum_{y\in \Phi^{(2)} \backslash \{x\}}g(x,y) \mathbf{1}_{\left\{x \overset{\Phi}{\leftrightarrow} y\right\}}
\end{eqnarray} 
The $1/2$ in front of the summation in (\ref{I2}) prevents us from considering a pair twice. Remark that we can consider here any type of signal (directional or not). 
 
\subsection{Expected value of $\mathcal{I}^{(1)}$ and $\mathcal{I}^{(2)}$.}

The next Theorem gives an exact integral expression to the expected value of the interference field generated by the singles and the pairs. The proof uses the Campbell-Little-Mecke formula, Lemma \ref{Lemma1}, and Theorem \ref{Percentage}.
\begin{theo}\label{Expected} For a PPP $\Phi$, the expected value of the \textit{interference field} generated by $\Phi^{(1)}$ and $\Phi^{(2)}$ is given by 
\begin{align}
\mathbb{E} \left[\mathcal{I}^{(1)}\right] &= (1-\delta) \int_{\mathbb{R}^2}\mathbb{E}\left[f(x)\right] \lambda dx, \label{expectedSingles}\\
\mathbb{E} \left[\mathcal{I}^{(2)}\right] &= \frac{1}{2} \int_{\mathbb{R}^2} \int_{\mathbb{R}^2}\mathbb{E}\left[g(x,y)\right]e^{-\lambda\pi |x-y|^2(2-\gamma)}\lambda dy \lambda dx. \label{expectedDoubles}
\end{align}
\end{theo}

\begin{proof}
Let us start with $\mathcal{I}^{(1)}$. We observe that, because the nearest neighbor always exists and is unique,
\begin{equation*}
\begin{split}
\mathcal{I}^{(1)} & = \sum_{x\in \Phi^{(1)}}f(x) \\
& = \sum_{x\in \Phi}f(x)\mathbf{1}_{\{x\in \Phi^{(1)}\}} \\
& = \sum_{x\in \Phi}f(x)\left( 1 - \mathbf{1}_{\{x\in \Phi^{(2)}\}} \right)
\end{split}
\end{equation*}
Thus, after applying the reduced Campbell-Little-Mecke formula and Slivnyak-Mecke's Theorem
\begin{equation*}
\begin{split}
\mathbb{E}\left[ \mathcal{I}^{(1)} \right] & = \mathbb{E}\left[ \sum_{x\in \Phi}f(x)\left( 1 - \mathbf{1}_{\{x\in \Phi^{(2)}\}} \right) \right] \\
& \stackrel{(a)}{=} \int_{\mathbb{R}^2} \mathbb{E} \left[ f(x) \right] \left( 1 - \mathbb{P}( x\in \Phi^{(2)}) \right) \lambda dx \\
& \stackrel{(b)}{=} (1-\delta) \int_{\mathbb{R}^2} \mathbb{E} \left[ f(x) \right]  \lambda dx,
\end{split}
\end{equation*}
where $(a)$ follows because $f(x)$ is independent of $\Phi$ and $(b)$ after Theorem \ref{Percentage}. Then we have the desired result for $\mathcal{I}^{(1)}$. 

For $\mathcal{I}^{(2)}$, we make the observation that 
\begin{equation*}
\sum_{x\in \Phi^{(2)}}\sum_{y\in \Phi^{(2)}\backslash \{x\}} g(x,y) \mathbf{1}_{\{x \stackrel{\Phi}{\leftrightarrow} y\}} = \sum_{x\in \Phi}\sum_{y\in \Phi\backslash \{x\}} g(x,y) \mathbf{1}_{\{x \stackrel{\Phi}{\leftrightarrow} y\}},
\end{equation*}
and iterating the reduced Campbell-Little-Mecke formula and Slivnyak-Mecke's Theorem, 
\begin{equation*}
\begin{split}
\mathbb{E}\left[ \mathcal{I}^{(2)}\right] & = \mathbb{E} \left[ \sum_{x\in \Phi}\sum_{y\in \Phi\backslash \{x\}} g(x,y) \mathbf{1}_{\{x \stackrel{\Phi}{\leftrightarrow} y\}} \right] \\
& = \int_{\mathbb{R}^2} \int_{\mathbb{R}^2} \mathbb{E} \left[ g(x,y) \mathbf{1}_{\{x \stackrel{\Phi}{\leftrightarrow} y\}} \right] \lambda dy \lambda dx  \\
& \stackrel{(c)}{=} \int_{\mathbb{R}^2} \int_{\mathbb{R}^2} \mathbb{E} \left[ g(x,y) \right] \mathbb{P} \left( x \stackrel{\Phi}{\leftrightarrow} y\right) \lambda dy \lambda dx \\
& \stackrel{(d)}{=} \int_{\mathbb{R}^2} \int_{\mathbb{R}^2} \mathbb{E} \left[ g(x,y) \right] e^{-\lambda \pi \|x-y\|^2(2-\gamma)} \lambda dy \lambda dx,  
\end{split}
\end{equation*}
where $(c)$ follows because $g(x,y)$ is independent of $\Phi$ and $(d)$ after Lemma \ref{Lemma1}.
\end{proof}
The expected value can be finite or infinite, depending on the choice of $f(x)$ and $g(x,y)$. Observe that for [NSC] and [PH] the expected interference has the same value.

\begin{cor}
\label{expected}
For a PPP $\Phi$, let $M^{(1)}$ and $M^{(2)}$ be the intensity measures of $\Phi^{(1)}$ and $\Phi^{(2)}$, respectively. Then, 
\begin{equation*}
\begin{split}
M^{(1)}(dx) & = (1-\delta)\lambda dx ,\\
M^{(2)}(dx) & = \ \delta\lambda dx. 
\end{split}
\end{equation*}
\end{cor}

\begin{proof}
Let $A$ be a regular subset of $\mathbb{R}^2$. For the choice $f(x)=\mathbf{1}^x_A$ (and $g(x,y)= \mathbf{1}^x_A 2$), the random variable $\mathcal{I}^{(1)}$ (and $\mathcal{I}^{(2)}$) counts the number of singles (pairs) within $A$. Applying directly the preceeding Theorem, and remarking that $\int_{\mathbb{R}^2} e^{-\lambda \pi \|x-y\|^2(2-\gamma)}\lambda dy = \delta$, for every $x\in \mathbb{R}^2$, we have the desired result. 
\end{proof}
The previous Corollary states that the intensities of $\Phi^{(1)}$ and $\Phi^{(2)}$ are $(1-\delta)\lambda$ and $\delta \lambda$, as stated in Section \ref{SecII}.

\subsection{Laplace functional of $\Phi^{(1)}$ and $\Phi^{(2)}$}

As a final discussion in this Section, we present our findings related to the LT of the interference from $\Phi^{(1)}$ and $\Phi^{(2)}$, when $\Phi$ is a PPP. Fix a measurable set $A\subset\mathbb{R}^2$ (window). Recall that $\Phi(A)$ denotes all the atoms of $\Phi$ inside $A$. We define the point processes  
\begin{equation}
\label{finiteWindow}
\begin{split}
& \Phi^{(1)}_A  =\left\{
\begin{tabular}{ l }
single atoms of $\Phi(A)$
\end{tabular}
\right\}\\
& \Phi^{(2)}_A =\left\{
\begin{tabular}{ l }
atoms of $\Phi(A)$ in \textit{MNNR} 
\end{tabular} 
\right\}\\
\end{split}
\end{equation}
where the MNNR have been considered only among the $\Phi$ atoms inside $A$. Consider a sequence of finite windows $(A_n)^\infty_{n=1}$ increasing to $\mathbb{R}^2$ in an appropriate sense (for example, $A_n=B(n,0)$). We have the following result.
\begin{theo}
\label{convergence}
Given a PPP $\Phi$, then, for $i=1,2$,
\begin{equation*}
\lim_{n\rightarrow \infty}\Phi^{(i)}_{A_n} \stackrel{(d)}{=} \Phi^{(i)},
\end{equation*}
where $\stackrel{(d)}{=}$ means equality in distribution.
\end{theo}
\begin{proof}
See Appendix B, available in the supplemental material.
\end{proof}
As convergence in distribution is equivalent to convergence of the LT \cite{KalRanMea}, the previous Theorem states that, for $A$ large enough, the LT of $\Phi^{(i)}_A$ approximates that one of $\Phi^{(i)}$. The benefit of this approach is that, for every finite window $A$, we can actually obtain an analytic representation for the LT of $\Phi^{(i)}_A$. As a sketch of the proof, fix a finite subset $A$. Conditioned on the number of atoms, these are i.i.d. uniformly distributed within $A$. Then, using the law of total probability, we can express the LT as an infinite sum of terms. The probability of a PPP having a fixed number of atoms within $A$ is known. Thus, we only have left to find expressions for the LT conditioned on the number of points inside $A$. For a finite number of different planar points $x_1,\ldots,x_n\in A$, define the function $H^{i,n}(x_1,\ldots,x_n)$ as the indicator function of the atom $x_i$ being in pair with another atom of the finite configuration $\{x_1,\ldots,x_n\}$ (Definition \ref{Single}). In the same fashion, define the function $I^{i,n}(x_1,\ldots,x_n)$ as the indicator function of the atom $x_i$ being single with respect to the finite configuration $\{x_1,\ldots,x_n\}$ (Definition \ref{Pair}). Let 
\begin{equation*}
\begin{split}
H^{(n)}&(x_1,\ldots,x_n) \\
& :=(H^{(1,n)}(x_1,\ldots,x_n),\ldots,H^{(n,n)}(x_1,\ldots,x_n)),\\
I^{(n)}&(x_1,\ldots,x_n) \\
& :=(I^{(1,n)}(x_1,\ldots,x_n),\ldots,I^{(n,n)}(x_1,\ldots,x_n)).
\end{split}
\end{equation*}
 We have the following result.

\begin{theo}[Laplace transform]
\label{LaplaceTransform}
Consider a PPP $\Phi$, with intensity $\lambda$, a regular subset $A\subset{\mathbb{R}^2}$, and a measurable function $f:\mathbb{R}^2\longrightarrow \mathbb{R}^+$. Let $F^{(n)}(x_1,\ldots,x_n):=(f(x_1),\ldots,f(x_n))$.

The LT of $\Phi^{(1)}_A$ is equal to 
\begin{equation}
\label{LPS}
\begin{split}
\mathbb{E} & \left( e^{-\sum_{x\in \Phi^{(1)}_A}f(x)} \right) \\
= & e^{-\lambda\mathcal{S}(A)} \Bigg(1 + \lambda\int_A e^{-f(x)} dx + \frac{\lambda^2}{2} \\
+ &\sum^ \infty_{n=3} \frac{\lambda^n}{n!}\int_A \ldots \int_A e^{-F^{(n)}(x_1,\ldots,x_n)\cdot H^{(n)}(x_1,\ldots,x_n)} dx_1 \ldots dx_n \Bigg)
\end{split}
\end{equation}
The LT of $\Phi^{(2)}_A$ is equal to
\begin{equation}
\label{LTD}
\begin{split}
 \mathbb{E} & \left(e^{-\sum_{x\in \Phi^{(2)}_A}f(x)}\right) \\
= & e^{-\lambda \mathcal{S}(A)} \Bigg(1 + \lambda \mathcal{S}(A) +\frac{\lambda^2}{2} \int_A \int_A e^{-(f(x)+f(y))}dy dx \\
+ &\sum^ \infty_{n=3} \frac{\lambda^n}{n!}\int_A\ldots \int_A e^{-F^{(n)}(x_1,\ldots,x_n)\cdot I^{(n)}(x_1,\ldots,x_n)}dx_1 \ldots dx_n \Bigg)
\end{split}
\end{equation}
\end{theo}
\textit{Implementation:} The MNNR was defined in a general way (see Section \ref{SecII}). Then, for every natural number $n$, it is easy to write a program/algorithm with input $(x_1,\ldots,x_n)$ and output $H^{(n)}(x_1,\ldots,x_n)$ (or $I^{(n)}(x_1,\ldots,x_n)$):
\begin{enumerate}
\item Define a $n\times n$ matrix $D=(d_{i,j})$, such that 
\begin{equation*}
d_{i,j}=\|x_i-x_j\|.
\end{equation*}
\item Choose a $n\times 1$ vector $v$, such that, for each $i=1,\ldots,n$,  
      \begin{equation*}
      v(i)=\argmin_{j\in \{1,\ldots,n\}\backslash \{i\}} d_{i,j}.
      \end{equation*}
\item Define another $n\times 1$ vector $u$, such that, for each $i=1,\ldots,n$
      \begin{itemize}
      \item \textbf{If} $i=v(v(i))$ (that is, $x_i$ and $x_{v(i)}$ are in MNNR), then $u(i)=1$.
      \item \textbf{Else} $u(i)=0$.
      \end{itemize}       
\item \textbf{Return} $u$.
\end{enumerate}
Fixed $f(x)$, $F^{(n)}(x_1,\ldots,x_n)$ is also known. Thus, for every natural number $n$, it is easy to set up a program that numerically approaches  
\begin{equation*}
\int_A \ldots \int_A e^{-F^{(n)}(x_1,\ldots,x_n)\cdot H^{(n)}(x_1,\ldots,x_n)} dx_1 \ldots dx_n 
\end{equation*}
and
\begin{equation*}
\int_A \ldots \int_A e^{-F^{(n)}(x_1,\ldots,x_n)\cdot I^{(n)}(x_1,\ldots,x_n)} dx_1 \ldots dx_n.
\end{equation*} 
However, it is clear that, as $n$ grows, computational time for $H^{(n)}(x_1,\ldots,x_n)$ and $I^{(n)}(x_1,\ldots,x_n)$ naturally increases. Since $n$-nested integral involving these functions needs to be computed, complicating the problem even more. As part of the current research, the authors work on the complexity reduction of the MNNR algorightm, and investigate the rate of convergence for the expressions in equations \eqref{LPS} and \eqref{LTD} (to avoid calculating an infinite number of terms).   

\section{The Superposition Model - Coverage Analysis}
\label{SecV}
As a consequence of the non-Poissonian behaviour of $\Phi^{(1)}$ and $\Phi^{(2)}$, a complete performance analysis of SINR related metrics is analytically challenging. This is due to the fact that the expressions presented for the LT are not numerically tractable. Thus, one cannot derive simple, analytic expressions for the coverage probability by LT methods, as shown in \cite{AndATract2011}. Instead, we use in this work the following model to approximate these metrics.
  
\subsection{Poisson Superposition Model}
To imitate the process of singles, we consider a PPP $\hat{\Phi}^{(1)}$, with parameter $(1-\delta)\lambda$. In this way, the new process of singles and $\Phi^{(1)}$ share the first moment (Corollary \ref{expected}).

To imitate the process of pairs, we also consider a PPP $\hat{\Phi}^{(2)}$, independent of $\hat{\Phi}^{(1)}$, with intensity $\frac{\delta}{2}\lambda$. We call the atoms of this process \textit{the parents}. We considere the process $\hat{\Phi}^{(2)}$ as independently marked. Each mark of a parent represents its pairing BS, \textit{the daughter}. The idea is that each couple $(\textit{parent},\textit{daughter})$ imitates a cooperating pair in MNNR. Let us consider $(Z_r)_{r>0}$ a family of independent, real r.v.s, independent also of $\hat{\Phi}^{(1)}$ and of $\hat{\Phi}^{(2)}$, where each $Z_r$ follows a Rice distribution, with parameter $(r,\alpha)$. If $Y$ is a random vector representing the Cartesian coordiantes of a parent, we define its mark by $Z_{\|Y\|}$.

To understand the choice for the marks, suppose that a BS is placed at the polar coordinates $(r,\theta)$, with $r>0$ and $\theta\in [0,2\pi)$ fixed (see Figure \ref{Zr}). Assume also that this BS belongs to a cooperating pair from the Nearest Neighbor model, and let us denote by $W$ the distance between the stations in pair. According to Theorem \ref{TheoNN2}, $W$ is Rayleigh distributed, with scale parameter $\alpha$. If $Z$ denotes the distance from the typical user to the second BS, the isotropy of the PPP implies that the distribution of $Z$ is independent of $\theta$. Moreover, we have the following result.

\begin{prop}\label{riceProp}
The r.v. $Z$ is Rice distributed, with parameters $(r,\alpha)$. The probability density function (PDF) of $Z$ is given by 
\begin{equation}\label{riceDen}
f(z|r)=\frac{z}{\alpha^2}e^{-\frac{z^2+r^2}{2\alpha^2}} I_0\left( \frac{zr}{\alpha^2} \right),
\end{equation}
where $I_0(x)$ is the modified Bessel function, of the first kind, with order zero. 
\end{prop}
The angular coordinate of a PPP atom is uniformly distributed in $[0,2\pi)$. Moreover, the Cartesian coordinates of a point around a center, with Rayleigh radial distance and uniform angle, are distributed as an independent Gaussian vector. Given this, Proposition \ref{riceProp} follows from \cite[Lem. 1]{ModPerfAfsh2016}.   
\begin{figure}
\centering
\includegraphics[width=0.2\textwidth]{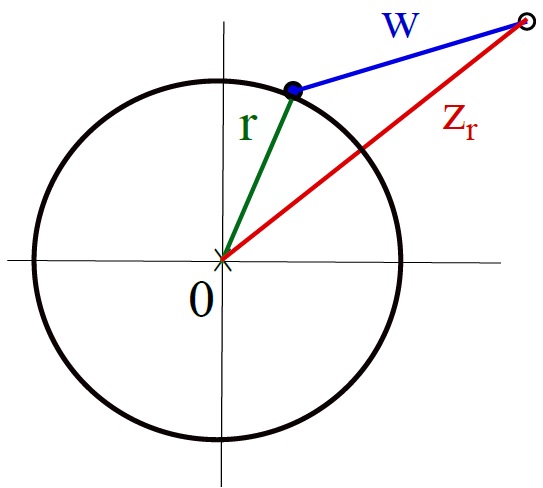}
\caption{Two cooperating BSs, where $r$ and $Z_r$ are their distances from the origin, and $W$ is the distance between them. \label{Zr}}
\end{figure}

\subsection{The Distribution of The Closest Distances} \label{nearestN}

Let $R_1$ and $R_ 2$ denote the r.v.s of the distances from the closest element of $\hat{\Phi}^{(1)}$ and $\hat{\Phi}^{(2)}$ to the origin, respectively. Denote also by $Z_2$ the mark of the parent at $R_2$. It is known that the r.v.s $R_1$ and $R_2$ are Rayleigh distributed \cite{AndATract2011}, with scale parameters $\xi$ and $\zeta$, where $\xi:=((1-\delta)2\lambda \pi)^{-1/2}$ and $\zeta:=(\delta \lambda \pi)^{-1/2}$. By definition, $R_2$ and $Z_2$ are not mutually independent, but we can derive their joint PDF.

\begin{lem}\label{denR2ZR2}
The joint PDF of the r.v. $(R_2,Z_2)$ is given by 
\begin{equation} \label{jointDensityFunction}
\begin{split}
f(r,z) & = \frac{rz}{(\alpha \zeta)^2} e^{-\frac{r^2}{2}\left(\frac{1}{\alpha^2}+\frac{1}{\zeta^2} \right)-\frac{-z^2}{2\zeta^2}}I_0\left( \frac{rz}{\zeta^2} \right).
\end{split} 
\end{equation}
Furthermore, the r.v. $Z_2$ is Rayleigh distributed, with scale parameter $( \alpha^2+\zeta^2)^{-1/2}$.
\end{lem}

\begin{proof}
See Appendix C, available in the supplemental material.
\end{proof}

To make a complete analysis of the coverage probability, we make use of the distribution of the random vector $(R_1,R_2,Z_2)$. Because $R_1$ is independent of $(R_2,Z_2)$, the joint PDF is the product of the PDF of $R_1$ with the joint PDF of $(R_2,Z_2)$. 

\subsection{Interference Field}

It is clear from the definition of the marks of $\Phi^{(2)}$ that, for the superposition model, we deal only with onmidirectional BSs. For $r>0$, denote by
\begin{subequations}\label{LPfunctions}
\begin{align}
& \mathcal{L}_{\tilde{f}}(s;r) := \mathbb{E}\left[ e^{-s \tilde{f}(r)}\right], \label{LPfunctionsa} \\
& \mathcal{L}_{\tilde{g}}(s;r,\rho) := \mathbb{E}\left[ e^{-s  \tilde{g}(r,Z_r)} \mathbf{1}_{\{Z_r>\rho\}}\right], \label{LPfunctionsb}
\end{align}
\end{subequations}
the LT of the signal generated by a single BS, and the LT of the signal generated by a cooperation pair, given that the radius of the daughter is larger than $\rho\geq 0$. When $\rho=0$, $\mathcal{L}_g(s;r,0)$ will be denoted just by $\mathcal{L}_g(s;r)$. For example, if we take $f(r)$ as in equation \eqref{signalsin}, we get  
\begin{equation}\label{LPsingles}
\mathcal{L}_f(s;r) = \frac{r^\beta}{sp+r^\beta }. 
\end{equation}
Recall that $p \frac{h_r} {r^\beta}$ and $p \frac{h_z} {z^\beta}$ are independent, exponential r.v.s, with parameter $\frac{r^\beta}{p}$ and $\frac{z^\beta}{p}$. In Table \ref{table} we find expressions for $\mathbb{E}[e^{-s\tilde{g}(r,z)}]$ in the \textbf{[NSC]}, \textbf{[OFF]}, and \textbf{[MAX]} cases. By remarking that 
\begin{equation*}
\mathcal{L}_{\tilde{g}}(s;r,\rho) = \mathbb{E}\left[\mathbb{E}\left[ e^{-s\tilde{g}(r,Z_r)} \Big|Z_r\right] \mathbf{1}_{\{Z_r>\rho\}}  \right],
\end{equation*}
we get analytical expressions for $\mathcal{L}_{\tilde{g}}(s;r)$ in the \textbf{[NSC]}, \textbf{[OFF]}, and \textbf{[MAX]}. For example, in the \textbf{[NSC]} we have that
\begin{equation*}
\begin{split}
\mathcal{L}_{\tilde{g}}(s;r,\rho) 
& = \frac{r^\beta}{sp+r^\beta} \int^\infty_\rho \frac{z^\beta}{sp+z^\beta} f(z|r)dr,
\end{split}
\end{equation*}
where $f(z|r)$ is the density function of the Rice r.v. $Z_r$ (see equation \eqref{riceDen}). For the more general distribution described by equation \eqref{TAIL}, it is also possible to give analytical formulas similar to $\mathcal{L}_{\tilde{g}}(s;r,\rho)$. The \textbf{[PH]} case is more complicated (see \cite[Lem. 3]{BacAStoGeo2015} for $cos(\theta_r-\theta_z)=1$). \\

We consider the interference fields generated by all the elements of $\hat{\Phi}^{(1)}$ and $\hat{\Phi}^{(2)}$ outside the radius $\rho\geq 0$
\begin{subequations} \label{interference}
\begin{align}
\hat{\mathcal{I}}^{(1)}(\rho) & = \sum_{x \in \hat{\Phi}^{(1)}, \|x\|>\rho} \tilde{f}(\|x\|), \label{interference2a} \\
\hat{\mathcal{I}}^{(2)}(\rho) & = \sum_{\stackrel{y \in \hat{\Phi}^{(2)}}{ \|y\|>\rho,Z_{\|y\|}>\rho}} \tilde{g}(\|y\|,Z_{\|y\|}). \label{interference2b}
\end{align}
\end{subequations}
When $\rho=0$, they are just denoted by $\hat{\mathcal{I}}^{(1)}$ and $\hat{\mathcal{I}}^{(2)}$. The total interference generated outside possibly different radii for the two processes, i.e. $\rho_1\geq 0$ and $\rho_2 \geq  0$ is
\begin{equation}\label{totalInt}
\hat{\mathcal{I}}(\rho_1,\rho_2):=\hat{\mathcal{I}}^{(1)}(\rho_1)+\hat{\mathcal{I}}^{(2)}(\rho_2).
\end{equation}
When $\rho_1=\rho_1=0$, we write only $\hat{\mathcal{I}}$.

The next Lemma is a well known result giving analytical representations to the LT of the PPP Interference fields \cite{BacBlaVol1}.

\begin{lem}\label{LTexpr}
The LTs of $\hat{\mathcal{I}}^{(1)}(\rho)$ and $\hat{\mathcal{I}}^{(2)}(\rho)$, denoted by $\mathcal{L}_{\hat{\mathcal{I}}^{(1)}}(s;\rho)$ and $\mathcal{L}_{\hat{\mathcal{I}}^{(2)}}(s;\rho)$, are given by   
\begin{subequations}\label{LPempty}
\begin{align}
& \mathcal{L}_{\hat{\mathcal{I}}^{(1)}} (s;\rho) = e^{-\lambda 2 \pi (1-\delta)\int^\infty_\rho \left(1-\mathcal{L}_f(s;r) \right)r dr}, \label{LPemptya}\\
& \mathcal{L}_{\hat{\mathcal{I}}^{(2)}} (s;\rho) = e^{- \pi \lambda \delta\int^\infty_{\rho}\left(1-\mathcal{L}_g(s;r,\rho)\right)r dr} \label{LPemptyb}.
\end{align}
\end{subequations}
\end{lem}
The Lemma uses the Poisson properties of $\hat{\Phi}^{(1)}$ and $\hat{\Phi}^{(2)}$. The expressions given in equations \eqref{LPempty} are the tools which allow us to make an entire analysis of the coverage probability.

As an example, if we replace equation \eqref{LPsingles} in equation \eqref{LPemptya}, for $\rho=0$ we get the analytical representation \cite{AndATract2011}
\begin{equation}\label{LTSinglesEx1}
\begin{split}
\mathcal{L}_{\hat{\mathcal{I}}^{(1)}}(s) &  = e^{-\frac{\lambda (1-\delta) 2\pi^2 (s p)^{2/\beta}}{\beta} csc\left(\frac{2 \pi}{\beta}\right)}, \\
\end{split}
\end{equation} 
where $csc(z)$ is the cosecant function. In the same fashion, it is possible to obtain expressions for $\mathcal{L}_{\hat{\mathcal{I}}^{(1)}} (s;\rho)$ and $\mathcal{L}_{\hat{\mathcal{I}}^{(2)}} (s;\rho)$. 

\subsection{Coverage Probability}

We can now make use of the PPP superposition model to evaluate the performance of the different cooperation (or coordination) types proposed above. The beneficial signal, received at the typical user from a single BS or a pair, will be denoted by $\tilde{f}(r)$ and $\tilde{g}(r,z)$, respectively. These may not be the same functions modeling the interference the typical user receives from other BSs. This is explained by the fact that the interference is the sum of the signals other BSs generate for their own serving users who are not located at the Cartesian origin. 

We consider two scenarios of user-to-BS association: 
\subsubsection{Fixed Single Transmitter}

Let us suppose that there is one BS serving the typical user, whose distance to the origin is fixed and known $r_0>0$. Moreover, it serves the typical user independently of the atoms from $\hat{\Phi}^{(1)}$ and $\hat{\Phi}^{(2)}$. Then the signal emitted to the typical user is $\tilde{f}(r_0)$, and the Signal-to-Interference-plus-Noise-Ratio (SINR) at the typical user is defined by 
\begin{equation}
\mathrm{SINR}:=\frac{\tilde{f}(r_0)}{\sigma^2+\hat{\mathcal{I}}}, 
\end{equation} 
where $\sigma^2$ is the additive Gaussian noise power at the receiver and $\hat{\mathcal{I}}$ is the total interference power (see equation \eqref{totalInt}).  

\begin{prop}
Suppose $\tilde{f}(r_0)$ as in \eqref{signalsin}. Then, the success probability is given by the expression 
\begin{equation}\label{SINRProbEF}
\mathbb{P}\left( \mathrm{SINR} >T \right) = e^{-\frac{T \sigma^2 r^\beta_0}{p} } \mathcal{L}_{\hat{\mathcal{I}}^{(1)}}\left(\frac{T r^\beta_0 }{p}\right)\mathcal{L}_{\hat{\mathcal{I}}^{(2)}}\left(\frac{T r^\beta_0}{p}\right).
\end{equation}
\end{prop}
The last proposition allows us to evaluate the SINR directly with the help of equations $\eqref{LPemptya}$ and $\eqref{LPemptyb}$ for $\rho=0$.

\subsection{Closest Transmitter from $\hat{\Phi}^{(1)}$ or $\hat{\Phi}^{(2)}$ (and his daughter)}

We consider that the typical user is connected to the BS at $R_1$ (see subsection \ref{nearestN}), or to the cooperating cluster (parent,daughter) at $(R_2,Z_2)$. The previous association depends on which one of them is closer to the typical user. If $R_1<\min\{R_2,Z_2\}$, the single BS at $R_1$ serves the typical user, and it emits the signal $\tilde{f}(R_1)$. In the opposite case, if $R_2\leq R_1$ or if $Z_2 \leq R_1$, the cooperating pair at $(R_2,Z_2)$ serves the user, and it emits the signal $\tilde{g}(R_2,Z_2)$. All the BSs not serving the typical user generate interference. Thus, 
\begin{equation}\label{SINR}
\mathrm{SINR} := 
\begin{cases}
\frac{\tilde{f}\left( R_1 \right)}{\sigma^2+\hat{\mathcal{I}}(R_1,R_1)} \ \ ; \ \ R_1<\min\{R_2,Z_2\}, \\
\frac{\tilde{g}\left(R_2,Z_2 \right)}{\sigma^2+\hat{\mathcal{I}}(R_2,R_2)} \ \ ; \ \ R_2<\min\{R_1,Z_2\},\\ 
\frac{\tilde{g}\left(R_2,Z_2 \right)}{\sigma^2+\hat{\mathcal{I}}(Z_2,R_2)} \ \ ; \ \ Z_2<\min\{R_1,R_2\}.\\ 
\end{cases}
\end{equation}
From equation \eqref{interference2b}, recall that once a \textit{parent} generates interference, its respective \textit{daughter} does it along with it. For the first term of the preceding equation, $\hat{\mathcal{I}}(R_1,R_1)$ considers that all the singles and \textit{parents} lying outside $R_1$ generate interference. For the the second term we use a similar argument. For the third one, the argument is a little bit more delicate. The r.v. $\hat{\mathcal{I}}(Z_2,R_2)$ considers that all the singles lie outside the radius $Z_2$, and all of them generate interference. Nevertheless, only the \textit{parents} outside $R_2$ generate interference (the parent associated to $R_2$ lies outside $Z_2$). Note that, the way this user-to-BS-association is defined, for the three cases, this is the only way to assure that all the BSs not serving the typical user generate interference.   
\begin{prop}
\label{CovProbNN}
Suppose $f(r)$ and $g(r,z)$ follow equations \eqref{signalsin} and \eqref{TAIL}. Then, there exist explicit functions $G:[0,\infty) \rightarrow \mathbb{R}^+$ and $H,K:[0,\infty)\times [0,\infty) \rightarrow \mathbb{R}^+$ such that 

\begin{equation*}
\label{CV2}
\begin{split}
\mathbb{P} & \left( \mathrm{SINR} >T \right) = \mathbb{E}[G(R_1)]+\mathbb{E}[H(R_2,Z_2)]+\mathbb{E}[K(R_2,Z_2)]. 
\end{split}
\end{equation*}
\end{prop}

\begin{proof}
See Appendix D, available in the supplemental material.
\end{proof}

The expressions for $G(r)$, $H(r,z)$, and $K(r,z)$ are given by  
\begin{equation*}
\begin{split}
& G(r) \ \ \ =  \tilde{G}(r) \hat{G}(r),  \\
& H(r,z)  =  \tilde{H}(r,z) \hat{H}(r,z),  \\ 
& K(r,z)  =  \tilde{K}(r,z) \hat{K}(r,z),
\end{split}
\end{equation*} 
where 
\begin{equation*}
\begin{split}
& \tilde{G}(r) \ \ \  := 1-F_{R_2}(r)-F_{Z_2}(r)+F_{R_2,Z_2}(r,r), \\
& \tilde{H}(r,z) := (1-F_{R_1}(r))\mathbf{1}_{\{z>r\}}, \\
&  \tilde{K}(r,z) := (1-F_{R_1}(z))\mathbf{1}_{\{r>z\}},
\end{split}
\end{equation*}
and
\begin{equation*}
\begin{split}
\hat{G} & (r) := e^{\frac{- T r^\beta}{p} \sigma^2 } \mathcal{L}_{\hat{\mathcal{I}}^{(1)}}\Bigg(\frac{Tr^\beta}{p};r \Bigg) \mathcal{L}_{\hat{\mathcal{I}}^{(2)}}\Bigg(\frac{Tr^\beta}{p};r \Bigg), \\
\hat{H} & (r,z) := \\
& \sum^n_{i=1}c_i \big(r,z\big) e^{-T d_i(r,z)\sigma^2} \mathcal{L}_{\hat{\mathcal{I}}^{(1)}}\big(Td_i(r,z);r\big)\mathcal{L}_{\hat{\mathcal{I}}^{(2)}}\big(Td_i(r,z);r\big) \\
\hat{K} & (r,z) :=  \\ 
& \sum^n_{i=1}c_i \big(r,z\big) e^{-T d_i(r,z)\sigma^2} \mathcal{L}_{\hat{\mathcal{I}}^{(1)}}\big(Td_i(r,z);z\big)\mathcal{L}_{\hat{\mathcal{I}}^{(2)}}\big(Td_i(r,z);r\big). 
\end{split}
\end{equation*}
We can find expressions for the deterministic functions $\mathcal{L}_{\hat{\mathcal{I}}^{(1)}}(s;\rho)$ and $\mathcal{L}_{\hat{\mathcal{I}}^{(2)}}(s;\rho)$ in equation \eqref{LPempty}, and the functions $c_i(r,z)$ and $d_i(r,z)$ are those from equation \eqref{TAIL}. Finally, the functions $F_{R_1}(r)$, $F_{R_2}(r)$, and $F_{Z_2}(r)$ are the CDF of the random variables $R_1$, $R_2$, $Z_2$, which are Rayleigh distributed (see Section \ref{nearestN}), and $F_{R_2,Z_2}(r,z)$ is the CDF of the random vector $(R_2,Z_2)$, which can be obtained with equation \eqref{jointDensityFunction}. 

\textit{Remark:} we can either calculate the expression of the coverage probability from equation \eqref{CV2} via Monte Carlo simulations (because we know the distribution of $R_1$, $R_2$, and $(R_2,Z_2)$), or via numerical integration, using the formulas  
\begin{equation*}
\begin{split}
& \mathbb{E}[G(R_1)] \ \ \ \ = \int^\infty_0 G(r)f_{R_1}(r) dr \\
& \mathbb{E}[H(R_2,Z_2)] = \int^\infty_0 \int^\infty_0 H(r,z)f_{R_2,Z_2}(r,z)dzdr \\
& \mathbb{E}[K(R_2,Z_2)] = \int^\infty_0 \int^\infty_0 K(r,z)f_{R_2,Z_2}(r,z)dzdr,
\end{split}
\end{equation*}
where $f_{R_1}(r)$ and $f_{R_2,Z_2}(r,z)$ are the density functions of $R_1$ and $(R_2,Z_2)$ (again, see equation \eqref{jointDensityFunction}). 

\section{Numerical Evaluation} 
\label{SecVI}
We consider a density for the BSs $\lambda=0.25$ [$km^2$], which corresponds to an average closest distance of $(2\sqrt{\lambda})^{-1}=1$ [km] between stations. We also consider that the power is $p=1$ [$Watt$].

We first illustrate the validity of the expressions in Theorem \ref{Expected}. Specifically, we compare the expressions in equations \eqref{expectedSingles} and \eqref{expectedDoubles} with simulations.  

\subsection{Expected value of the Interference field}
Given a fixed $\beta>2$, define the random field $f(x)=\frac{h_{\|x\|}}{\|x\|^\beta}\mathbf{1}_{\{\|x\|>R\}}$, where $R$ is a positive number and the family $(h_r)_{r>0}$ is defined in Section \ref{SecIII}. The indicator function serves to calculate the interference generated by the singles, outside a ball centred at $0$ and radius $R$. With the aid of $f(x)$, define $\mathcal{I}^{(1)}$ as in equation \eqref{I1}. Using Theorem \ref{Expected}, the numerical evaluation of the expected value of $\mathcal{I}^{(1)}$ is given in figure \ref{ExpectedPlots}. The expression in \eqref{expectedSingles} gives almost identical results with the simulations. 

Similarly, with the aid of the random field $g(x,y)\textbf{1}_{\{\|x\|,\|y\|>R\}}$, define $\mathcal{I}^{(2)}$ as in equation \eqref{I2}. For the numerical evaluation, we consider the two cases $\textbf{[NSC]}$ and $\textbf{[MAX]}$. The interference from $\textbf{[MAX]}$ is always smaller than that one from $\textbf{[NSC]}$, since it is received only from one of the two BSs of each pair, while the other is silent. Figure \ref{ExpectedPlots} shows that the numerical evaluation of the expression in \eqref{expectedDoubles} gives almost identical results with the simulations. Remark also that, for $\beta=4$, the two scenarios do not numerically defer much.

For the coverage probability analysis, we evaluate only the noiseless scenario $\mathbb{P}(\mathrm{SIR}>T)$ (with $\sigma^2=0$). We compare the $\mathrm{SIR}$ coverage performance from the Nearest Neighbor and the superposition models against the model without cooperation \cite{AndATract2011}.  We consider both cases $(a)$ with fixed transmitter, and $(b)$ where the association is done with the (almost) closest cluster, as in (\ref{SINR}). In this second case, for the Nearest Neighbor model, the user-cluster association is done differently than in the superposition model, as follows. The typical user is served by the closest BS of the original point process $\Phi$, and by its mutually nearest neighbor, if one exists. The cooperative signals are those proposed in \eqref{coopFunction}.

\begin{figure}[htbp] 
\centering
\subfigure[]{\includegraphics[trim = 30mm 95mm 35mm 95mm, clip,width=0.33\textwidth]{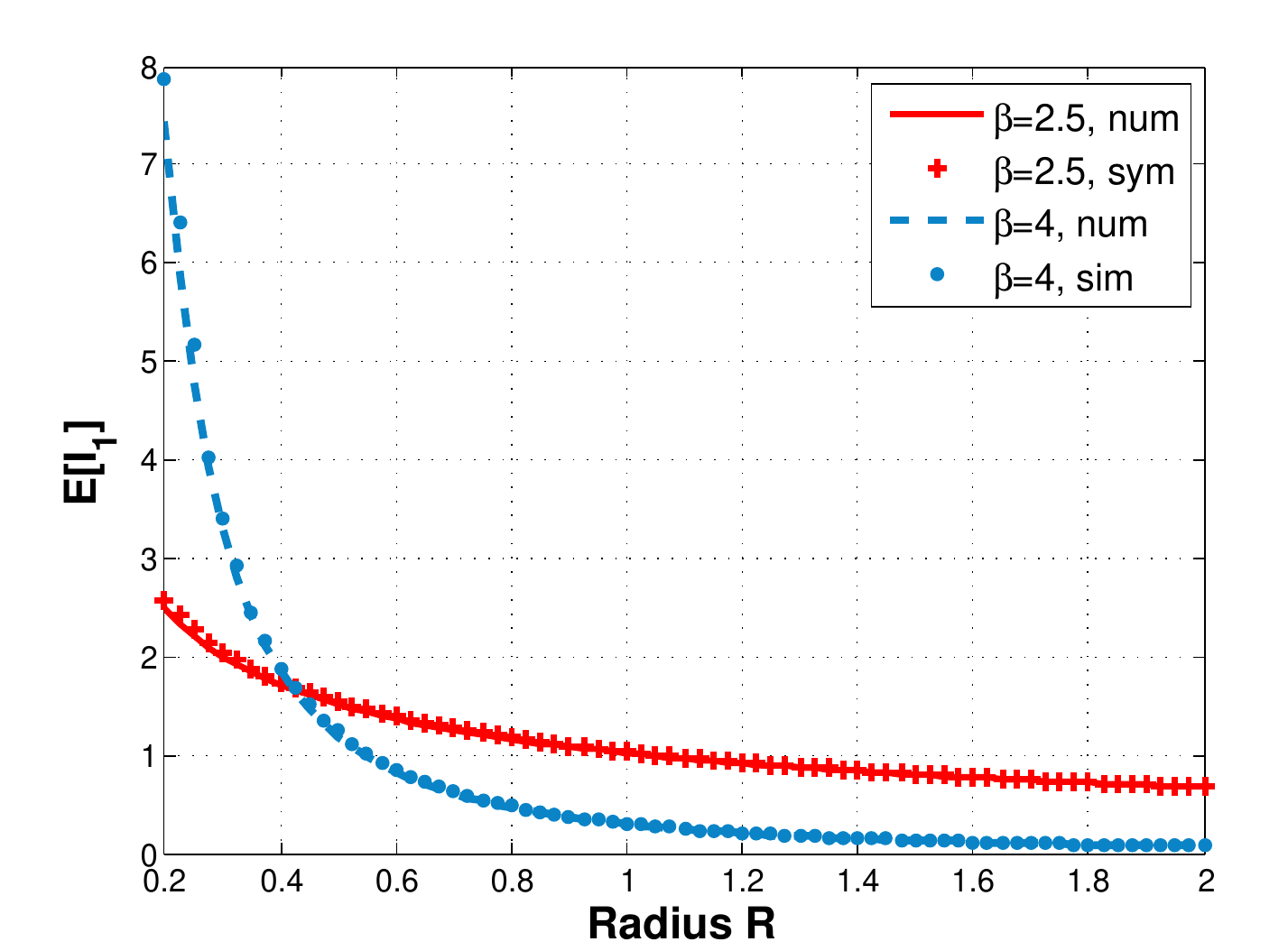}}
\subfigure[]{\includegraphics[trim = 0mm 0mm 0mm 0mm, clip,width=0.35\textwidth]{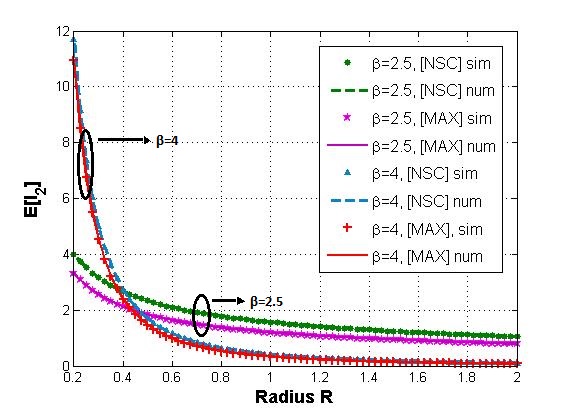}} 
\caption{(a) Expected value of the interference generated by the single atoms, outside a radius $R$, (b) and by the cooperative pairs, outside a radius $R$. }   \label{ExpectedPlots}
\end{figure}

\subsection{Closeness of the approximation by the PPP superposition}

We compare in Fig. \ref{NNVsSPCom} the coverage probabilities, over the threshold $T$, for the Nearest Neighbor model and the superposition model, in both association cases. As we can see, the curves are very close in both cases. For the \textit{"closest" transmission cluster} the difference is more evident, because on the one hand the superposition model does not take into account the repulsion between clusters (singles or pairs), and on the other hand the association of a cluster to the user as done in \eqref{SINR} for the superposition model, sometimes misses the actual closest daughter to the origin (which is not necessarily the one at $Z_2$). This never happens the way we choose the closest cluster in the Nearest Neighbor model. Hence, the approximative model underestimates the coverage benefits in the closest cluster association. 

\begin{figure}[htbp] 
\centering
\subfigure[]{\includegraphics[trim = 30mm 95mm 35mm 95mm, width=0.34\textwidth]{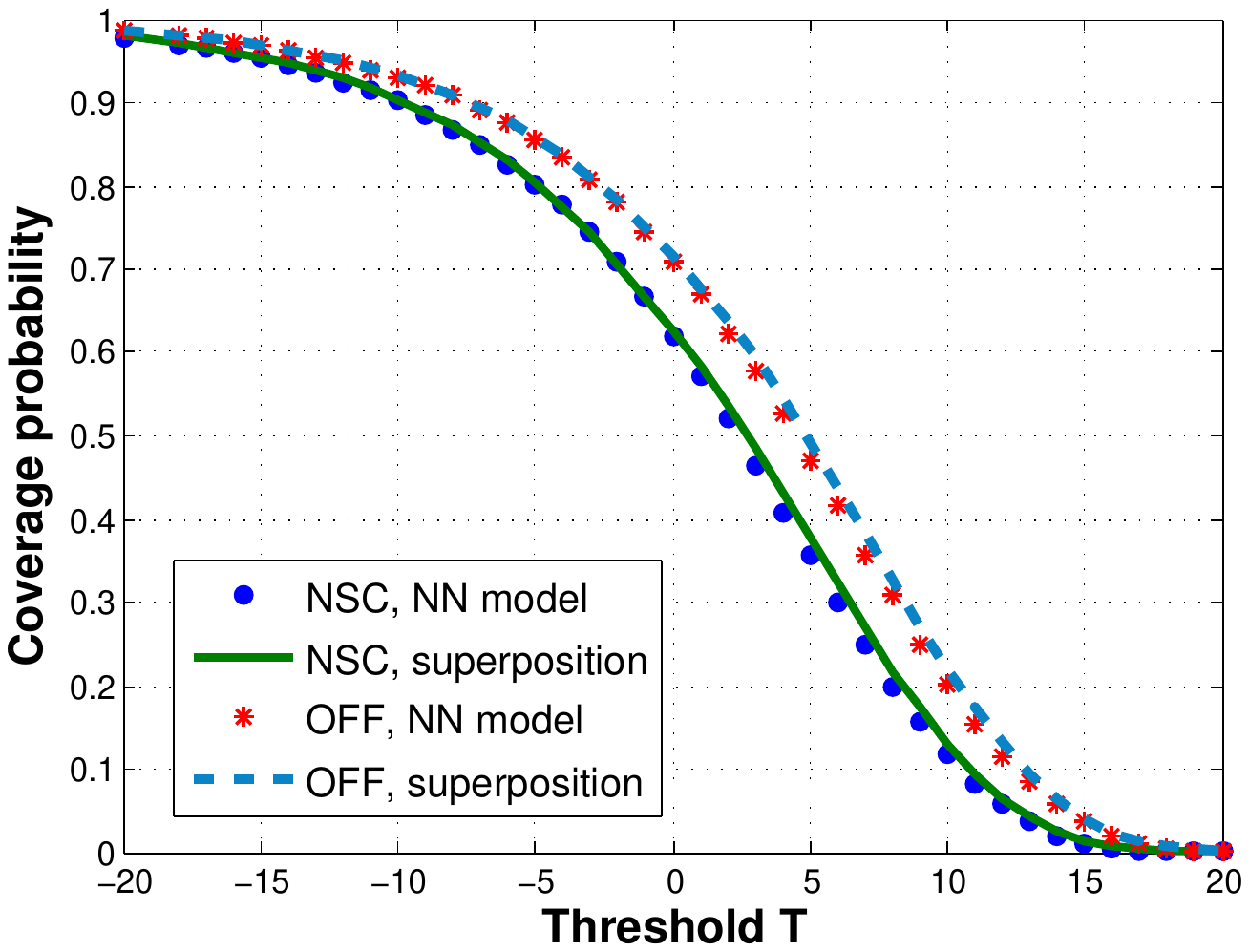}}
\subfigure[]{\includegraphics[trim = 30mm 95mm 35mm 95mm,width=0.34\textwidth]{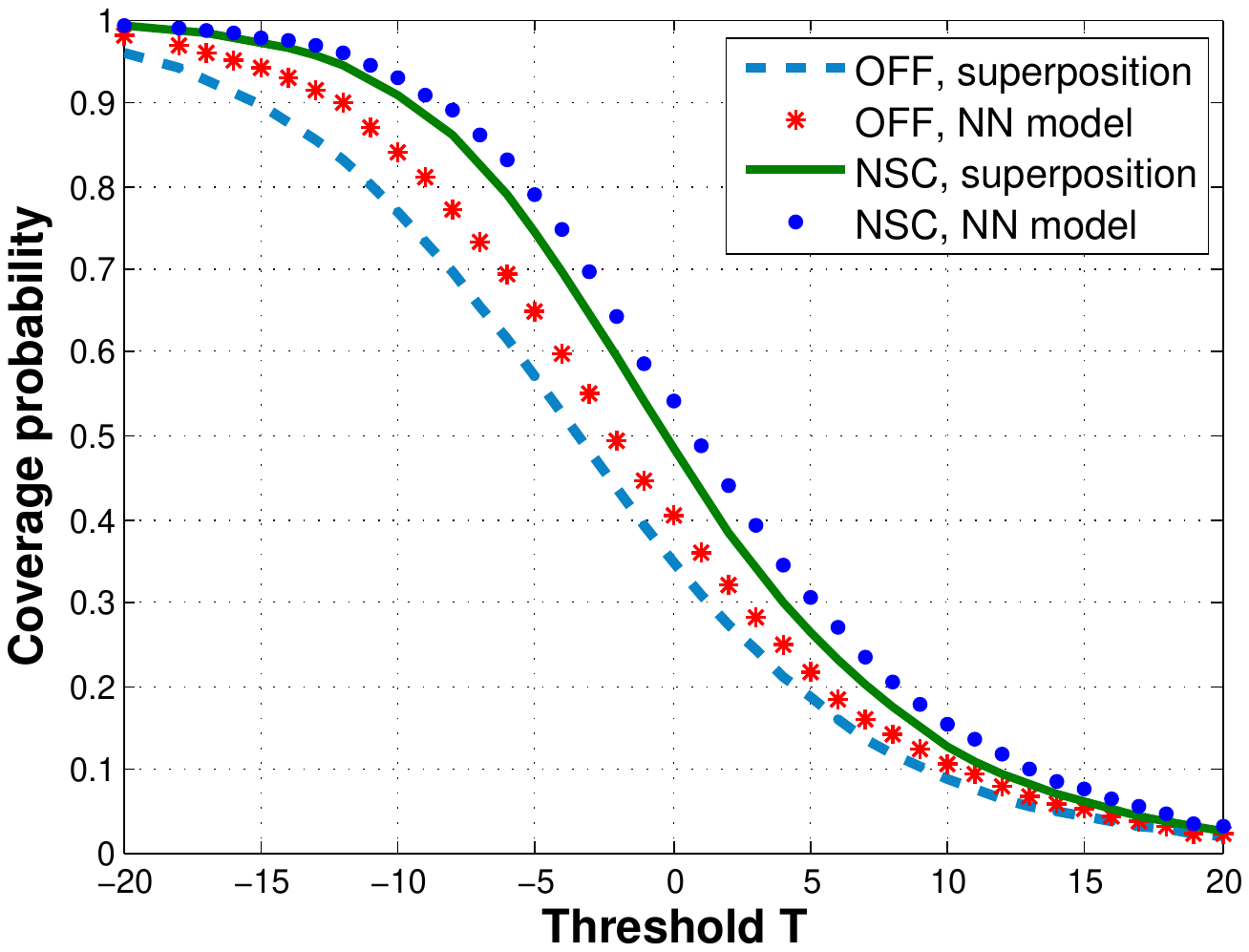}} 
\caption{Closeness of the approximation between the superposition and the Nearest Neighbor model, $\beta=3$. (a) Fixed transmitter and (b) closest transmitter.} \label{NNVsSPCom} 
\end{figure}

\subsection{Validity of the numerical analysis}
In Figure \ref{AnalVsSim}, we compare the plots of the coverage probability from the numerical integration, against simulations, of the analytic formula presented in Proposition \ref{CovProbNN}. 
As we can see, they 
fit perfectly, both for larger values of $\beta$, like $\beta=4$, and for critical ones, like $\beta=2.5$.

\begin{figure}[htbp] 
\centering
\subfigure[]{\includegraphics[trim = 30mm 95mm 35mm 95mm,width=0.34\textwidth]{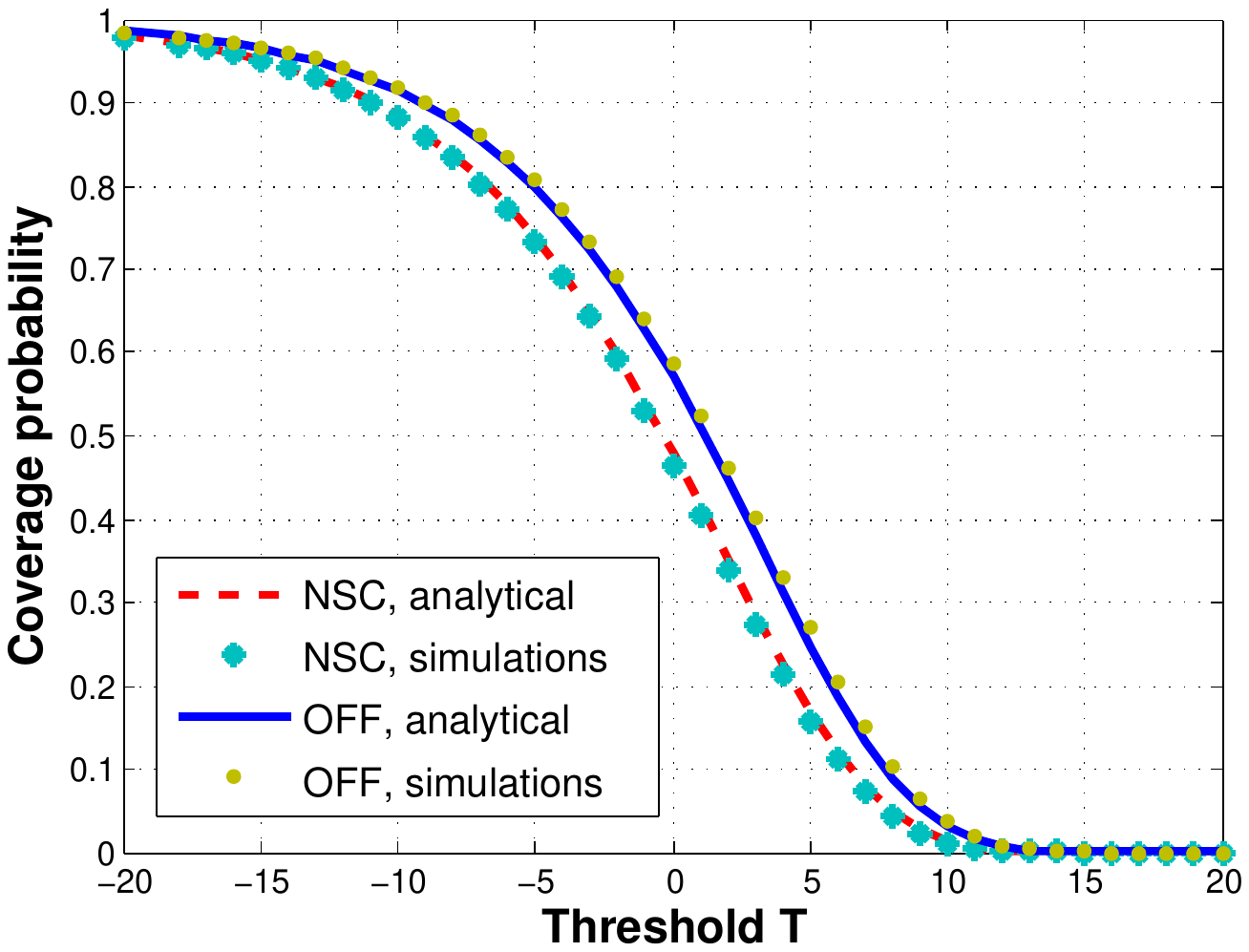}}
\subfigure[]{\includegraphics[trim = 30mm 95mm 35mm 95mm,width=0.34\textwidth]{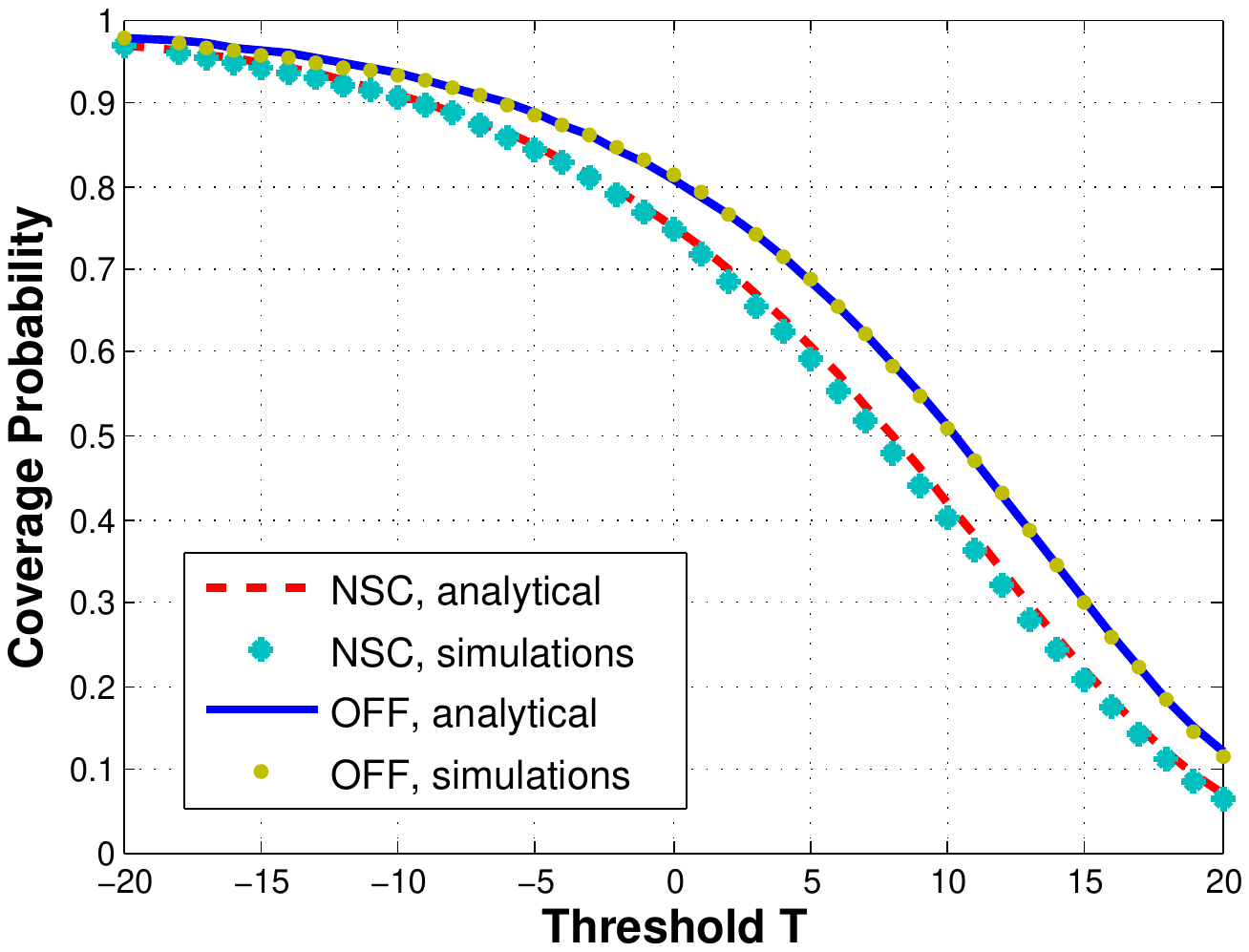}} 
\caption{Validity of the analysis for the superposition model for the fixed single transmitter. (a) $\beta=2.5$ (b) $\beta=4$.} \label{AnalVsSim}  
\end{figure}

\subsection{Cooperation gains}

\begin{figure}[htbp] 
\centering
\subfigure[]{\includegraphics[trim = 30mm 95mm 35mm 95mm,width=0.34\textwidth]{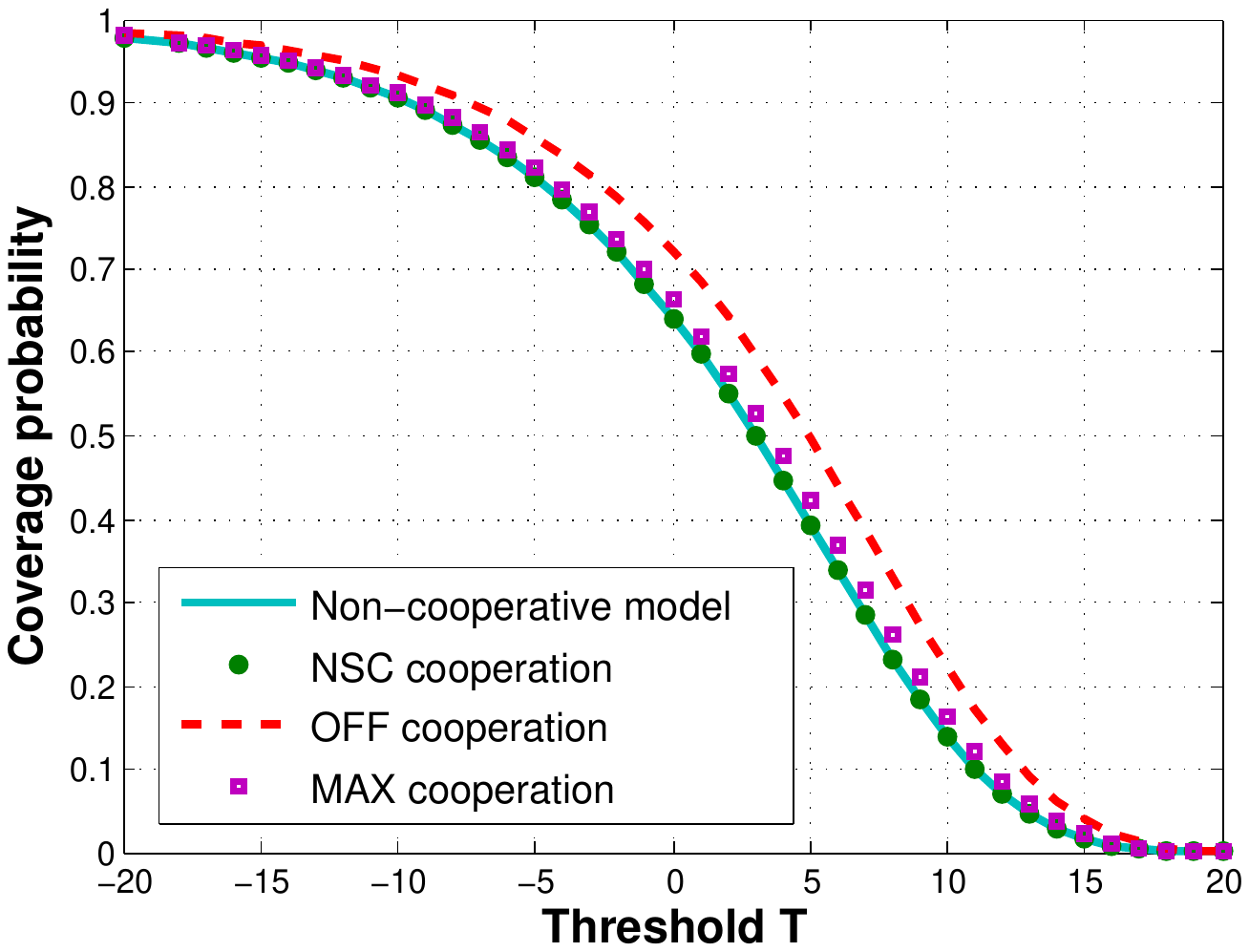}}
\subfigure[]{\includegraphics[trim = 30mm 95mm 35mm 95mm,width=0.34\textwidth]{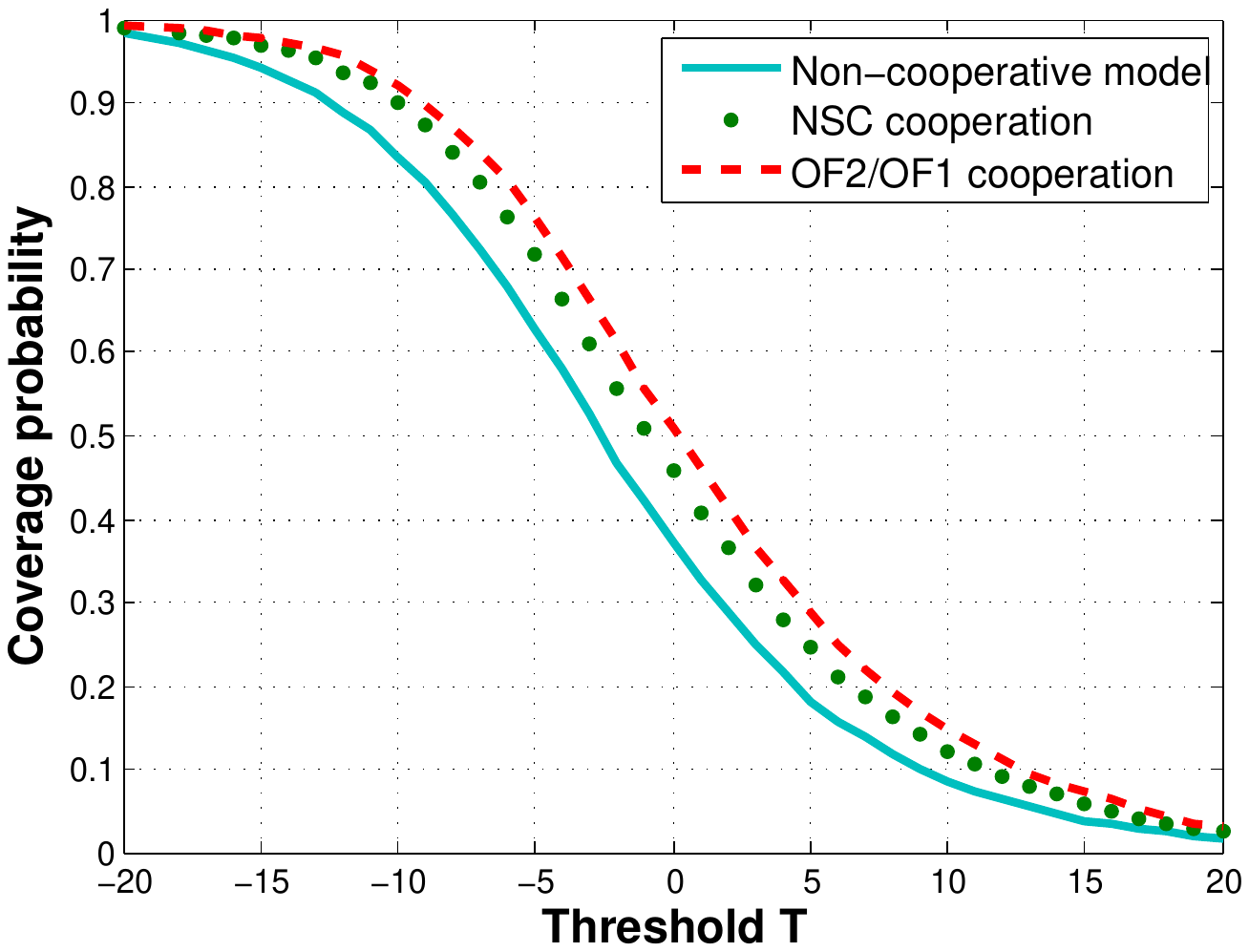}} 
\caption{Coverage gains from Nearest Neighbor cooperation compared to no cooperation, $\beta=3$. (a) Fixed transmitter and (b) closest transmitter.} \label{CoopVsNoCoppNN}
\end{figure}

The possible coverage gains, compared to the non-cooperative network, in the case of an association with a fixed transmitter, are shown in Fig. \ref{CoopVsNoCoppNN}(a). As a first remark, for the fixed association, the \textbf{[NSC]} case for the Nearest Neighbor model and the non-cooperative model are practically the same. The coverage probability in the \textbf{[MAX]} case is close to the coverage probability in the \textbf{[NSC]} case. This suggests that the strongest signal in each cooperating pair influences interference the most. For the \textbf{[OFF]} case there is a $10\%$ benefit compared to the non-cooperative case, in the largest part of the domain in $T$.

The gains are also evaluated in the case of association with the closest cluster. For the SINR, let us call \textbf{[MAX/OFF]} the case where the closest cluster emits a signal to the typical user according to \textbf{[MAX]}, i.e. only the max signal is sent, while the pairs generate interference, according to \textbf{[OFF]}. The idea is that when all network pairs choose \textbf{[MAX]} cooperation for their own users, this choice of one-station-out-of-two is random for the typical user point of view. This \textbf{[MAX/OFF]} case shows a $15\%$ absolute gain from the non cooperative case, which is around $9\%$ for the \textbf{[NSC]} (see Fig. \ref{CoopVsNoCoppNN}(b)). This gain is almost equal with the dynamic clustering in \cite{BacAStoGeo2015}.

\section{Model Benefits and Future Research}
\label{SecVII}

The static grouping model presented in this paper has the following network benefits. 

\begin{itemize}
\item By definition, the MNNR reduces the generated interference.
\item The percentage of stations that are in cooperative pairs and the percentage of those left single (both for the PPP and the hexagonal grid model) show that the MNNR is a reasonable grouping strategy.  
\item Our approach can be applied to many cooperation variations, ranging from simple coordination of the BSs in group, to fully cooperative transmission using knowledge of the channel states. 
\end{itemize}
The mathematical innovations are the following.

\begin{itemize}
\item The derivation of structural properties for $\Phi^{(1)}$ and $\Phi^{(2)}$ by classical Stochastic Geometry tools comes naturally.
\item Two repulsive point processes can be constructed in a natural way. 
\item Based on simple geometrical concepts, the MNNR can be easily implemented in any programming language. This simplifies the numerical evaluation of the system.
\item The superposition approach makes possible a complete analysis of the coverage probability, which can be extended to other performance metric, such as the throughput.        
\item The superposition is equal in distribution to a Gauss-Poisson process. However, the mark related to each pair is chosen in an original way that allows the derivation of simple formulas, that cannot be obtained by directly using a Gauss-Poisson.
\end{itemize} 

As stated in Theorem \ref{convergence}, the resulting edge effects of the MNNR are negligible for a PPP. This makes possible a finite window analysis for the interference generated by the singles and pairs. Part of the current research of the authors is the convergence rate for this result, as well as bounds and the respective convergence rate for Theorem \ref{LaplaceTransform}.

The static cluster methodology proposed is based on the Euclidean distance between BSs, and fixes groups of singles and pairs over time. This approach doest not allow for flexibility in the way the groups are created. We could imagine introducing different methods, for the formation of the desirable clusters, that take into account the availability of each BS as well as the Euclidean disatnce among them. Analysis and applications of this type of cooperation can be considered as future extensions of the MNNR. 
 
\section{Conclusions}
\label{SecVIII} 

The MNNR is a reasonable methodology to define single BSs and cooperative pairs. In spite of the analytical difficulties (due to overlapping discs with different radii), it is possible to use an approximate model and provide a complete analysis of SINR-related metrics. Moreover, it could be possible to generalise the idea and include clusters of size greater than two. The analysis should remain the same and similar results could be derived via Monte Carlo simulations.

The coverage benefits of the MNNR, with respect to the non-cooperative case, can reach a $15\%$ of absolute gain, although around $32\%$ of stations is single and do not cooperate. Similar gains can be achieved by some dynamic clustering methodologies. For the MNNR, this is impressive, considering that only $62\%$ of the BSs cooperate. Different kinds of cooperative signals reported different coverage benefits. Thus, cooperation benefits fundamentally depend on the choice of the grouping method, the allowed maximum cluster size, as well as the appropriate cooperation signals. This works provides an important step towards resolving this complex problem.

%
%
%
%


%

\ifCLASSOPTIONcompsoc
 \section*{Acknowledgment}
The authors would like to thank Aur\'elien Vasseur for the help on the proof of Theorem \ref{convergence}, and for the suggested related literature.

\ifCLASSOPTIONcaptionsoff
  \newpage
\fi



%
\bibliographystyle{unsrt}
\footnotesize
\bibliography{FixClust}

%

\begin{IEEEbiography}[{\includegraphics[width=1in,height=1.25in,clip]{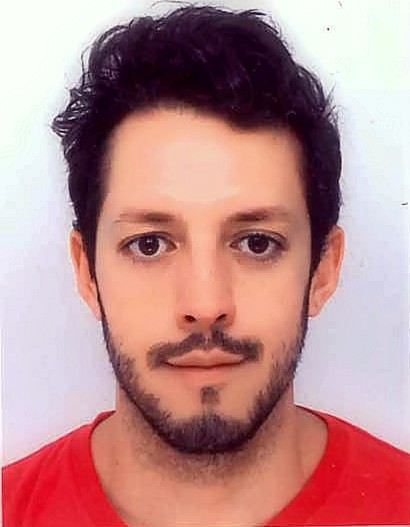}}]{Luis David \'Alvarez Corrales}
is a Ph.D. student at the Network and Computer Science Department, Telecom ParisTech. He received the Master's degree in Probability and Stochastic Models 
from the University Pierre et Marie Curie (UPMC), laboratory LPMA. He has further completed the Master’s degree in Mathematical Finance and the Bachelor's degree in Mathematics at the National Autonomous University of Mexico (UNAM). In Mexico, he has worked in CRM analysis, energy regulation, and insurance regulation (Solvency II). His research lies in the area of applied probability, stochastic modeling, and stochastic geometry applications for the performance evaluation of telecommunication networks. 

\end{IEEEbiography}

\begin{IEEEbiography}[{\includegraphics[width=1in,height=1.25in,clip]{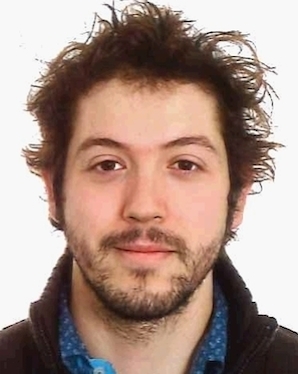}}]{Anastasios Giovanidis}
is a permanent researcher of the French National Center for Scientific Research (CNRS). After spending three years with Telecom ParisTech, laboratory CNRS-LTCI, he is now affiliated with the University Pierre et Marie Curie (UPMC), laboratory CNRS-LIP6. He has been a postdoctoral fellow, first with the Zuse Institute Berlin, Germany and later with INRIA, Paris, France. He received the Dr.-Ing. degree in Mobile Communications from the Technical University of Berlin, Germany and the Diploma in Electrical and Computer Engineering from the National Technical University of Athens, Greece. His research lies in the area of performance evaluation and optimization of telecommunication networks, with emphasis in queuing and stochastic geometry applications.
\end{IEEEbiography}

\begin{IEEEbiography}[{\includegraphics[width=1in,height=1.25in,clip]{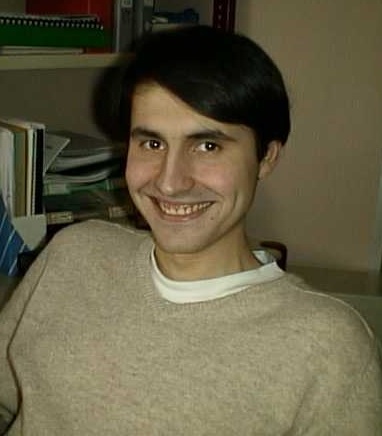}}]{Philippe Martins}
is Professor in the networking and computer science department, at Télécom ParisTech (Paris, France). He is an IEEE senior member. His main research interests lie in protocol design and performance 
evaluation of mobile networks. He worked on stochastic geometry models for dimensioning and planning. He is also investigating on coverage modelling and energy harvesting using models based on simplicial homology.\end{IEEEbiography}

\begin{IEEEbiography}[{\includegraphics[width=1in,height=1.25in,clip]{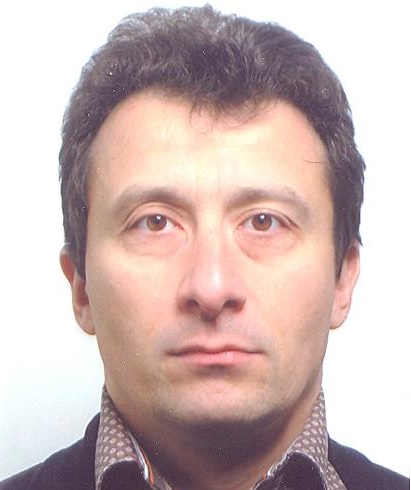}}]{Laurent Decreusefond}
is a former student of Ecole Normale Supérieure de Cachan. He obtained his Ph.D. degree in Mathematics in 1994 from Telecom ParisTech and his Habilitation in 2001. He is currently a Professor in the Network and Computer Science Department, at Telecom ParisTech. His main fields of interest are the Malliavin calculus, the stochastic analysis of long range dependent processes, random geometry and topology and their applications. With P. Moyal, he co-authored a book about the stochastic modelling of telecommunication
\end{IEEEbiography}





\clearpage

\newpage 

\appendices
\section*{Appendix A (Supplemental Material)}
\label{AppendixA}

\subsection*{Proof of Proposition \ref{VoronoiPerc}}
Denote by $(R,\Theta)$ the polar coordinates of the closest $\Phi$ atom from the typical user. The r.v. $R$ is Rayleigh distributed, with scale parameter $(\lambda 2 \pi)^{-1/2}$. Because of the isotropy of a stationary PPP, $R$ and $\Theta$ are independent r. v. and $\Theta$ is uniformly distributed over $[0,2\pi)$. Therefore, its density function is given by 
\begin{equation*}
f_{(R,\Theta)}(r,\theta)=\lambda r e^{-\lambda\pi r^2}\mathbf{1}_{\{r>0\}}\mathbf{1}_{\{\theta\in [0,2\pi)\}}
\end{equation*} 
Then, 
\begin{equation*}
\begin{split}
\mathbb{P}(0 & \curvearrowright \Phi^{(2)}) = \mathbb{E}\left( \mathbb{P}(0\curvearrowright \Phi^{(2)}|R,\Theta) \right) \\
& = \int^\infty_0 \int^{2\pi}_0 \mathbb{P}(0\curvearrowright \Phi^{(2)}|R=r,\Theta=\theta)\lambda r e^{-\lambda\pi r^2} d\theta dr
\end{split}
\end{equation*}
Fix a realisation $\phi$. Denote by $x$ and $y$ two different atoms from $\phi$, whose polar coordinates are $(r,\theta)$ and $(s,\varphi)$, respectively. If $\rho$ denotes the Euclidean distance between $x_1$ and $x_2$, then, 
\begin{equation*}
\rho^2=r^2+s^2-2rscos(\theta-\varphi)
\end{equation*}
If we suppose that $x$ is the nearest neighbor atom from $\phi$ to the origin, then, the atoms $x$ and $y$ are in MNNR iff 
\begin{equation}
D(x,y):=(B(x,\rho)\cup B(y,\rho))\backslash B(0,r)
\end{equation}
is empty of atoms from $\phi \backslash \{x,y\}$. Denote by $F(r,s,\theta,\varphi)$ the Euclidean surface of $D(x,y)$, the empty space function of a PPP implies that
\begin{equation*}
\begin{split}
\mathbb{P}(0\curvearrowright & \Phi^{(2)}|R=r,\Theta=\theta) \\
& = \int^\infty_0 \int^{2\pi}_0 e^{-\lambda F(r,s,\theta,\varphi)}sdsd\varphi,
\end{split}
\end{equation*}
and therefore, 
\begin{equation*}
\begin{split}
\mathbb{P}(0 & \curvearrowright \Phi^{(2)}) \\
& = \lambda^2 \int^\infty_0 \int^{2\pi}_0 \int^\infty_0 \int^{2\pi}_0 e^{-\lambda F(r,s,\theta,\varphi)-\lambda\pi r^2} rs dsd\varphi d\theta dr 
\end{split}
\end{equation*}
In some cases, it is actually possible to find explicit values for $F(r,s,\theta,\varphi)$. For example, the case $\rho \geq 2 r$ implies that $B(0,r) \subset B(x,\rho)$. Thus, 
\begin{equation*}
\mathcal{S}(D(x,y))=\mathcal{S}(B(y,\rho)\backslash B(x,\rho) )+\mathcal{S}(B(x,\rho))-\mathcal{S}(B(0,r))
\end{equation*}
Then, we have that
\begin{equation*}
\mathcal{S}(B(y,\rho)\backslash B(x,\rho))=\pi \rho^2(1-\gamma),
\end{equation*}
and
\begin{equation*}
\mathcal{S}(B(x,\rho))-\mathcal{S}(B(0,r))=\pi \rho^2-\pi r^2.
\end{equation*}
Unfortunately, in other cases is arduous to obtain $F(r,s,\theta,\varphi)$. 

\section*{Appendix B (Supplemental Material)}
\label{AppendixB}

\subsection*{Proof of Theorem \ref{convergence}}
For a natural number $n$, denote $B_n:=B(0,n)$ and fix $\Phi_n:=\Phi^{(1)}_{B_n}$, as done in equation \eqref{finiteWindow}. We will prove that, for every compact subset $E$ of $\mathbb{R}^2$,
\begin{itemize}
\item[(i)] \begin{equation*}
\lim_{n\rightarrow \infty} \mathbb{P}(\Phi_n(E)=0)=\mathbb{P}(\Phi^{(1)}(E)=0)
\end{equation*}
\item[(ii)]
\begin{equation*}
\limsup_{n\rightarrow \infty} \mathbb{P}(\Phi_n(E)\leq 1)\geq \mathbb{P}(\Phi^{(1)}(E)\leq 1)
\end{equation*} 
\item[(iii)] 
\begin{equation*}
\lim_{t \nearrow \infty}\limsup_{n\rightarrow \infty}\mathbb{P}(\Phi_n(E)>t)=0
\end{equation*}
\end{itemize}
The previous being equivalent to convergence in distribution of the sequence of point processes $(\Phi_n)$ to $\Phi^{(1)}$ \cite{KalRanMea}. 
Fix a compact $E\subset \mathbb{R}^2$. Let us start to prove $(i)$. Being $\Phi_n$ a thinning of the PPP $\Phi$,
\begin{equation*}
\begin{split}
\mathbb{P}(\Phi_n(E)=0) = & e^{-\lambda \mathcal{S}(E)} \\
&+\mathbb{P}(\Phi_n(E)=0,\Phi(E)>0) \\
= & e^{-\lambda \mathcal{S}(E)} \\
 &+ \mathbb{P}(\Phi_n(E)=0,\Phi^{(1)}(E)=0,\Phi(E)>0) \\
 & + \mathbb{P}(\Phi_n(E)=0,\Phi^{(1)}(E)>0,\Phi(E)>0)
\end{split}
\end{equation*}
Given that the compact subset $E$ is fixed, take a natural number $n_1$ such that 
\begin{equation*}
n_1>3 sup_{y \in E} \|y\|
\end{equation*} 
and such that $E\subset B_n$, for every $n>n_1$. Therefore, for every atom belonging to $\Phi^{(1)}$, but not to $\Phi_n$, the distance to its nearest neighbor must exceed $\frac{2}{3}sup_{y \in E} \|y\|$. Thus, there exits a constant $C_1>0$ such that, for every $n>n_1$,
\begin{equation*}
\begin{split}
\mathbb{P}(\Phi_n(E)=0,\Phi^{(1)}(E)>0,\Phi(E)>0) \leq e^{-\lambda \pi C_1 n^2 }
\end{split}
\end{equation*}
In the same fashion, 
\begin{equation*}
\begin{split}
\mathbb{P}(\Phi^{(1)}(E)=0)  & = e^{-\lambda \mathcal{S}(E)}  \\
& +\mathbb{P}(\Phi_n(E)=0,\Phi^{(1)}(E)=0,\Phi(E)>0)\\
& + \mathbb{P}(\Phi^{(1)}(E)=0,\Phi_n(E)>0,\Phi(E)>0)
\end{split}
\end{equation*}
and there must exists a natural number $n_2$, and a constant $C_2>0$ such that, for every $n>n_2$,
\begin{equation*}
\mathbb{P}(\Phi^{(1)}(E)=0,\Phi_n(E)>0,\Phi(E)>0) \leq  e^{-\lambda \pi C_2 n^2 }
\end{equation*}
Take $N=max\{n_1,n_2\}$, then, for every $n>N$,
\begin{equation*}
\begin{split}
|\mathbb{P}(\Phi_n(E)=0)-\mathbb{P}(\Phi^{(1)}(E)=0)| \leq e^{-\lambda \pi C_1 n^2 }+e^{-\lambda \pi C_2 n^2 }
\end{split}
\end{equation*}
We conclude that 
\begin{equation*}
\begin{split}
\lim_{n\rightarrow \infty} \mathbb{P}(\Phi_n(E)=0)=\mathbb{P}(\Phi^{(1)}(E)=0)
\end{split}
\end{equation*}
To prove $(ii)$, remark that 
\begin{equation*}
\begin{split}
\mathbb{P}(\Phi_n(E)=1) = & \ \ \ \mathbb{P}(\Phi_n(E)=1,\Phi^{(1)}(E)\neq 1) \\
& + \mathbb{P}(\Phi_n(E)=1,\Phi^{(1)}(E)=1) \\
\mathbb{P}(\Phi^{(1)}(E)=1) = & \ \ \ \mathbb{P}(\Phi^{(1)}(E)=1,\Phi_n(E)\neq 1) \\
& + \mathbb{P}(\Phi^{(1)}(E)=1,\Phi_n(E)=1) \\
\end{split}
\end{equation*}
hence
\begin{equation*}
\begin{split}
|\mathbb{P}&(\Phi_n(E)=1)-\mathbb{P}(\Phi^{(1)}(E)=1)| \\
\leq & | \ \mathbb{P}(\Phi_n(E)=1,\Phi^{(1)}(E)\neq 1) \\
& - \mathbb{P}(\Phi^{(1)}(E)=1,\Phi_n(E)\neq 1)|  
\end{split}
\end{equation*}
In the same way as we did before, we can prove that 
\begin{equation*}
\begin{split}
\lim_{n\rightarrow \infty} | & \ \mathbb{P}(\Phi_n(E)=1,\Phi^{(1)}(E)\neq 1) \\
& - \mathbb{P}(\Phi^{(1)}(E)=1,\Phi_n(E)\neq 1)|=0,
\end{split}
\end{equation*}
and this leads to 
\begin{equation*}
\lim_{n\rightarrow \infty}\mathbb{P}(\Phi_n(E)=1)=\mathbb{P}(\Phi^{(1)}(E)=1)
\end{equation*}
Finally, we prove $(iii)$. Being $\Phi_n$ a thinning of the PPP $\Phi$, 
\begin{equation*}
\begin{split}
\mathbb{P}(\Phi_n(E)>t) & \leq \mathbb{P}(\Phi(E)>t) \\
& \leq \frac{\mathbb{E} \Phi(E)}{t}\\
& = \frac{\lambda \mathcal{S}(E)}{t} 
\end{split}
\end{equation*} 
which goes to zero, as $t \nearrow \infty$.

Take a sequence $(A_n)$ of compact sets. To conclude that $(\Phi^{(1)}_{A_n})$ converges in distribution to $\Phi^{(1))}$, then it must fulfil that, for every natural number $m$, there exits another natural number $N$ such that, for every $n>N$, then, $B_m \subset A_n$. We can prove that $(\Phi^{(2)}_{A_n})$ converges in distribution to $\Phi^{(2)}$.
\section*{Appendix C (Supplemental Material)}
\subsection*{Proof of Proposition \ref{denR2ZR2}.} 
Denote by $A$ and $B$ the Cartesian coordinates of the nearest parent to the typical user and his daughter, respectively, and also denote by $(R_2,\Theta)$ and $(Z_2,\Psi)$ their respective polar coordinates. Define $C:=A-B$ and denote its polar coordinates by $(W,\Omega)$ (see Figure \ref{prueba}). The random variables $R_2$ and $W$ are Rayleigh distributed, with scale parameters $\zeta$ and $\alpha$, respectively. Moreover, the random angles $\Theta$, $\Psi$ and $\Omega$ are considered uniformly distributed over $[0,2\pi)$, to preserve the isotropy in the PPP case. Also, the random variables $R_2$, $\Theta$, $W$, and $\Omega$ are independent between them, as in the PPP case. 

Our first goal is to find the joint distribution of the random vector $(R_2,Z_2)$ and, as a consequence, find also the distribution of $Z_2$. The cartesian coordinates of a point around a center, wich has rayleigh radial distance from the origin and uniform angle, are distibuted as an independent Gaussian vector \cite[pp. 276, Ex. 7b]{AfirstCourseRoss}. Hence, there exist independent random variables $A_x,A_y,C_x,C_y$, where $A_x,A_y$ are Normal distributed, with parameter $(0,\zeta^2)$, and $C_x,C_y$ are also Normal distributed, with parameters $(0,\alpha^2)$, and such that 
\begin{equation*}
\begin{split}
(A_x,A_y) \stackrel{d}{=} (R_2cos\Theta,R_2sin\Theta), \\
(C_x,C_y) \stackrel{d}{=} (Wcos\Omega,Wsin\Omega).
\end{split}
\end{equation*}
\begin{figure}[htbp]
\centering
\subfigure[]{\includegraphics[width=0.2\textwidth]{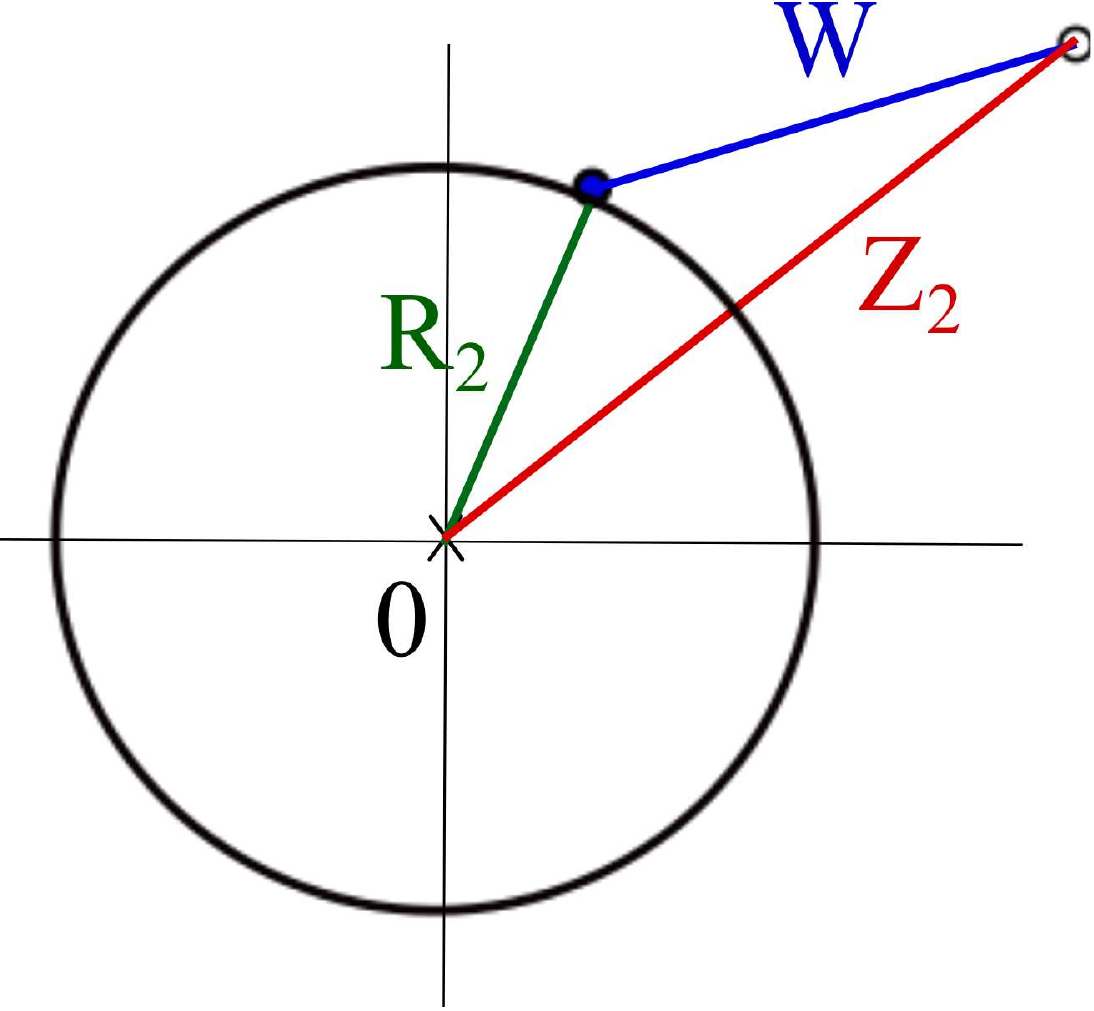}}
\caption{}\label{prueba}
\end{figure}
By definition, 
\begin{equation*}
\begin{split}
A_x & \stackrel{d}{=} R_2 cos\Theta, \ \ \ \  C_x \stackrel{d}{=} R_2cos\Theta-Z_2cos\Psi, \\
A_y & \stackrel{d}{=} R_ 2 sin\Theta, \ \ \ \ C_y \stackrel{d}{=} R_2sin\Theta-Z_2 sin\Psi.
\end{split}
\end{equation*}
The absolute value of the Jacobian of the above transformation is $R_2 Z_ 2$. Denote by  $f_{A_x,A_y,C_x,C_y}$ the joint PDF of $(A_x,A_y,C_x,C_y)$, if $f_{R,\Theta,Z,\Psi}$ denotes the joint PDF  of $(R_2,\Theta,Z_2,\Psi)$, then, by the change of variable Theorem \cite[pp. 274]{AfirstCourseRoss},  
\begin{equation*}
\begin{split}
f & _{R,\Theta,Z,\Psi} (r,\theta,z,\psi) \\
& = f_{A_x,A_y,C_x,C_y}(rcos\theta,rsin\theta,rcos\theta-zcos\psi,rsin\theta-zsin\psi)rz \\
& \stackrel{(a)}{=} \frac{rz}{(2\pi\alpha \zeta)^2}e^{-\left( \frac{ r^2cos^2\theta}{2\zeta^2}+\frac{ r^2sin^2\theta}{2\zeta^2}+\frac{ (rcos\theta-zcos\psi)^2}{2\alpha^2}+\frac{ (rsin\theta-zsin\psi)^2}{2\alpha^2} \right)} \\
& \stackrel{(b)}{=} \frac{rz}{(2\pi\alpha \zeta)^2}e^{-\left( \frac{r^2}{2} \left(\frac{1}{\alpha^2}+\frac{1}{\zeta^2}\right)
+\frac{ z^2}{2\alpha^2}-\frac{ rzcos(\theta-\psi)}{\alpha^2} \right)},
\end{split}
\end{equation*}
where $(a)$ comes from the formula of the distribution of independent Gaussian random variables, and $(b)$ follows from the trigonometric identities $cos^2\theta+sin^2\theta =1$ and $cos\theta cos\psi+sin\theta sin\psi =cos(\theta-\psi)$. To obtain the joint PDF of $(R_2,Z_2)$, denoted by $f_{R_2,Z_2}$, we integrate the previous expression over $[0,2\pi)\times [0,2\pi)$, with respect to the variables $\theta$ and $\psi$, 
\begin{equation*}
\begin{split}
f & _{R_2,Z_2} (r,z) \\
& \stackrel{(c)}{=} \frac{rz}{( \alpha \zeta)^2}e^{-\left( \frac{r^2}{2} \left(\frac{1}{\alpha^2}+\frac{1}{\zeta^2}\right)+\frac{z^2}{2\alpha^2} \right) } \frac{1}{2\pi}\int^{2\pi}_0  e^{\frac{ rzcos w}{\alpha^2} }dw \\
& \stackrel{(d)}{=} \frac{rz}{(\alpha \zeta)^2}e^{-\left( \frac{r^2}{2} \left(\frac{1}{\alpha^2}+\frac{1}{\zeta^2}\right)+\frac{z^2}{2\alpha^2} \right) } I_0\left(\frac{rz}{\alpha^2} \right),
\end{split}
\end{equation*}
where $(c)$ comes from the change of variable $w=\theta-\psi$ and $(d)$ follows because the integral representation $I_0(x)=\frac{1}{2\pi}\int^{2\pi}_0  e^{x cosw }dw$ \cite{ModBessEric}. Let us denote by $f_{Z_2}$ the PDF of the random variable $Z_2$ and by $\eta=\left(\frac{1}{\alpha^2}+\frac{1}{\zeta^2}\right)$. To obtain $f_{Z_2}$, we integrate over $[0,\infty)$ with respect to the variable $r$ the preceding equation
\begin{align*}
f_{Z_2}(z) & \stackrel{(e)}{=} \int^\infty_0 \frac{rz}{(\alpha\zeta)^2}e^{-\left(\frac{r^2}{2} \eta+\frac{z^2}{2\alpha^2}\right)}\sum^\infty_{n=0}\frac{(1/4)^n}{(n!)^2}\left( \frac{rz}{\alpha^2} \right)^{2n} dr \\
& =  \frac{z}{(\alpha\zeta)^2}e^{-\frac{z^2}{2\alpha^2}} \sum^\infty_{n=0} \frac{(1/4)^n}{(n!)^2}\left( \frac{z^2}{\alpha^4} \right)^{n}  \int^\infty_0 r^{2n} e^{-\frac{r^2}{2}  \eta} r dr \\
& \stackrel{(f)}{=} \frac{z}{(\alpha\zeta)^2\eta}e^{-\frac{z^2}{2\alpha^2}} \sum^\infty_{n=0} \frac{\left(\frac{z^2}{2 \alpha^4 \eta }\right)^n}{n!} \\
& = \frac{z}{(\alpha\zeta)^2\eta} e^{-\frac{z^2}{2\alpha^2} }e^{\frac{z^2}{2 \alpha^4\eta} } \\
& \stackrel{(g)}{=}  \frac{z}{\alpha^2+\zeta^2} e^{-\frac{z^2}{2(\alpha^2+\zeta^2) } }, \\
\end{align*}
where $(e)$ comes from the series representation $I_0(x)=\sum^\infty_{n=0}\frac{(1/4)^n}{(n!)^2}x^{2n}$ \cite{ModBessEric}, while $(f)$ follows after the formula
\begin{equation*}
\int^\infty_0 r^{2n} e^{-\frac{r^2}{2}\eta} r dr=\frac{2^n}{\eta^{n+1}}n!,
\end{equation*}
and $(g)$ after soma algebraic manipulations and from the definition of $\eta$.

\section*{Appendix D (Supplemental Material)}
\subsection*{Proof of Proposition \ref{CovProbNN} }
We split the proof in three parts.

\subsubsection*{A first expression}

The events $\{R_1<\min\{R_2,Z_2\}\}$, $\{R_2<\min\{R_1,Z_2\}\}$, and $\{Z_2<\min\{R_1,R_2\}\}$ are mutually independent, then, from equation \eqref{SINR},  
\begin{equation}
\label{CP2}
\begin{split}
\mathbb{P} & (\mathrm{SINR}>T) \\
 & = \mathbb{P}\left( \frac{\tilde{f}\left( R_1 \right)}{\sigma^2+\hat{\mathcal{I}}(R_1,R_1)}>T,R_1<\min\{R_2,Z_2\} \right) \\
& + \mathbb{P}\left( \frac{\tilde{g}\left( R_2,Z_2 \right)}{\sigma^2+\hat{\mathcal{I}}(R_2,R_2)}>T,R_2<\min\{R_1,Z_2\} \right) \\
& + \mathbb{P}\left( \frac{\tilde{g}\left( R_2,Z_2 \right)}{\sigma^2+\hat{\mathcal{I}}(Z_2,R_2)},Z_2<\min\{R_1,R_2\} \right) \\
& = \mathbb{E}\left[ \mathbf{1}_{ \left\{ \frac{\tilde{f}\left( R_1 \right)}{\sigma^2+\hat{\mathcal{I}}(R_1,R_1)}>T \right\}} \mathbf{1}_{ \left\{ R_1<\min\{R_2,Z_2\} \right\} } \right] \\
& + \mathbb{E}\left[ \mathbf{1}_{ \left\{ \frac{\tilde{g}\left( R_2,Z_2 \right)}{\sigma^2+\hat{\mathcal{I}}(R_2,R_2)}>T \right\} } \mathbf{1}_{ \left\{ R_2<\min\{R_1,Z_2\} \right\} } \right] \\
& + \mathbb{E}\left[ \mathbf{1}_{ \left\{ \frac{\tilde{g}\left( R_2,Z_2 \right)}{\sigma^2+\hat{\mathcal{I}}(Z_2,R_2)} >T \right\} } \mathbf{1}_{ \left\{ Z_2<\min\{R_1,R_2\} \right\} } \right]
\end{split} 
\end{equation}
For the first term we have that  
\begin{equation*}
\begin{split}
\mathbb{E} &\Bigg[  \mathbf{1}_{\Big\{\frac{\tilde{f}\left( R_1 \right)}{\sigma^2+\hat{\mathcal{I}}(R_1,R_1)}>T \Big\} }  \mathbf{1}_{\{R_1<\min{\{R_2,Z_2\} } \} }   \Bigg]  \\
&  = \mathbb{E}  \Bigg[ \mathbb{E}\Bigg[  \mathbf{1}_{\Big\{\frac{\tilde{f}\left( R_1 \right)}{\sigma^2+\hat{\mathcal{I}}(R_1,R_1)}>T \Big\} }  \mathbf{1}_{\{R_1<\min{\{R_2,Z_2\} } \} }\Big | R_1,R_2,Z_2 \Bigg] \Bigg]  \\
& \stackrel{(a)}{=} \mathbb{E} \Bigg[ \mathbf{1}_{\{R_1<\min{\{R_2,Z_2\} } \} } \mathbb{E}\Bigg[  \mathbf{1}_{\Big\{\frac{\tilde{f}\left( R_1 \right)}{\sigma^2+\hat{\mathcal{I}}(R_1,R_1)}>T \Big\} } \Big | R_1,R_2,Z_2 \Bigg] \Bigg]  \\
& \stackrel{(b)}{=} \mathbb{E} \Bigg[ \mathbf{1}_{\{R_1<\min{\{R_2,Z_2\} } \} } \mathbb{E}\Bigg[  \mathbf{1}_{\Big\{\frac{\tilde{f}\left( R_1 \right)}{\sigma^2+\hat{\mathcal{I}}(R_1,R_1)}>T \Big\} } \Big | R_1 \Bigg] \Bigg]  \\
& = \mathbb{E} \Bigg[ \mathbf{1}_{\{R_1<\min{\{R_2,Z_2\} } \} }  \mathbb{P}\Bigg( \frac{\tilde{f}\left( R_1 \right)}{\sigma^2+\hat{\mathcal{I}}(R_1,R_1)}>T  \Big | R_1 \Bigg) \Bigg],  
\end{split}
\end{equation*}
where $(a)$ comes from the properties of the conditional expectation and $(b)$ follows because the event $\left\{\frac{\tilde{f}\left( R_1 \right)}{\sigma^2+\hat{\mathcal{I}}(R_1,R_1)}>T\right\}$ is independent of $R_2$ and $Z_2$.

After a similar analysis for the two terms with the cooperative signal 
\begin{equation*}
\begin{split}
\mathbb{E} &\Bigg[  \mathbf{1}_{\Big\{\frac{\tilde{g}\left( R_2,Z_2 \right)}{\sigma^2+\hat{\mathcal{I}}(R_2,R_2)}>T \Big\} }  \mathbf{1}_{\{R_2<\min{\{R_1,Z_2\} } \} }   \Bigg]  \\
& = \mathbb{E} \Bigg[ \mathbf{1}_{\{R_2<\min{\{R_1,Z_2\} } \} }  \mathbb{P}\Bigg( \frac{\tilde{g}\left( R_2,Z_2 \right)}{\sigma^2+\hat{\mathcal{I}}(R_2,R_2)}>T  \Big | R_2,Z_2 \Bigg) \Bigg], \\
\mathbb{E} &\Bigg[  \mathbf{1}_{\Big\{\frac{\tilde{g}\left( R_2,Z_2 \right)}{\sigma^2+\hat{\mathcal{I}}(Z_2,R_2)}>T \Big\} }  \mathbf{1}_{\{Z_2<\min{\{R_1,R_2\} } \} }   \Bigg]  \\
& = \mathbb{E} \Bigg[ \mathbf{1}_{\{Z_2<\min{\{R_1,R_2\} } \} }  \mathbb{P}\Bigg( \frac{\tilde{g}\left( R_2,Z_2 \right)}{\sigma^2+\hat{\mathcal{I}}(Z_2,R_2)}>T  \Big | R_2,Z_2 \Bigg) \Bigg]  
\end{split}
\end{equation*}
\subsubsection*{Some functions}
Denote by 
\begin{equation*}
\begin{split}
& \hat{G}(r) := \mathbb{P}\Bigg( \frac{\tilde{f}\left( R_1 \right)}{\sigma^2+\hat{\mathcal{I}}(R_1,R_1)}>T \Big| R_1=r \Bigg) \\
& \hat{H}(r,z) := \mathbb{P}\Bigg( \frac{\tilde{g}\left(R_2,Z_2 \right)}{\sigma^2+\hat{\mathcal{I}}(R_2,R_2)}>T \Big|R_2=r,Z_2=z\Bigg) \\
& \hat{K}(r,z) := \mathbb{P}\Bigg( \frac{\tilde{g}\left(R_2,Z_2 \right)}{\sigma^2+\hat{\mathcal{I}}(Z_2,R_2)}>T \Big| R_2=r,Z_2=z \Bigg).
\end{split}
\end{equation*}
For a given $r>0$, because $R_1$ is independent from $\hat{\mathcal{I}}(R_1,R_1)$, 
\begin{equation*}
\hat{G}(r)=\mathbb{P}\left( \tilde{f}\left( r \right)>T(\sigma^2+\hat{\mathcal{I}}(r,r)) \right)
\end{equation*}
Consider $\tilde{f}(r)$ as in \eqref{signalsin}, then it follows an exponential distribution with parameter $\frac{r^\beta}{p}$. Since $\hat{\mathcal{I}}(r,r)$ is independent of $\tilde{f}(r)$,  
\begin{align*}\label{Fun1}
\hat{G}(r) & =  \mathbb{E}\Big[\mathbb{P}\Big(\tilde{f}(r)>T\left( \sigma^2+\hat{\mathcal{I}}(r,r) \right)\Big| \hat{\mathcal{I}}(r,r) \Big)\Big] \nonumber \\
& = e^{\frac{- T r^\beta}{p} \sigma^2 } \mathcal{L}_{\hat{\mathcal{I}}^{(1)}}\Bigg(\frac{Tr^\beta}{p};r \Bigg) \mathcal{L}_{\hat{\mathcal{I}}^{(2)}}\Bigg(\frac{Tr^\beta}{p};r \Bigg), 
\end{align*}
where the deterministic functions $\mathcal{L}_{\hat{\mathcal{I}}^{(1)}}(s;\rho)$ and $\mathcal{L}_{\hat{\mathcal{I}}^{(2)}}(s;\rho)$ are given by \eqref{LPempty}. 

In the same fashion, for $r>0$ and $z>0$, because $(R_2,Z_2)$ is independent of $\hat{\mathcal{I}}(R_2,R_2)$,
\begin{align*}
\hat{H}(r,z)=\mathbb{P} \left(\tilde{g} \left(r,z\right)>T\left( \sigma^2+\hat{\mathcal{I}}(r,r) \right)\right)
\end{align*}
Using the general expression in \eqref{TAIL} for $\tilde{g}(r,z)$,
\begin{align*}
& \hat{H}(r,z) = \\ 
& \sum^n_{i=1}c_i \big(r,z\big) e^{-T d_i(r,z)\sigma^2} \mathcal{L}_{\hat{\mathcal{I}}^{(1)}}\big(Td_i(r,z);r\big)\mathcal{L}_{\hat{\mathcal{I}}^{(2)}}\big(Td_i(r,z);r\big)
\end{align*}
We do the same to find and expression for $\hat{K}(r,z)$.
 
\subsubsection*{Final expression}
To complete the analysis, we need to find the coverage probability expressed in equation \eqref{CP2}, thus, we need expressions for   
\begin{equation*}
\begin{split}
\mathbb{E} & \left[\hat{G}(R_1)\textbf{1}_{\{R_1<\min\{R_2,Z_2\}\}}\right], \\
\mathbb{E} & \left[\hat{H}(R_2,Z_2)\textbf{1}_{\{R_2<\min\{R_1,Z_2\}\}}\right], \\
\mathbb{E} & \left[\hat{K}(R_2,Z_2)\textbf{1}_{\{Z_2<\min\{R_1,R_2\}\}}\right]. \\
\end{split}
\end{equation*}
Let us begin by the first one,
\begin{equation*}
\begin{split}
\mathbb{E} & \left[\hat{G}(R_1)\textbf{1}_{\{R_1<\min{R_2,Z_2}\}}\right] \\
& \stackrel{(a)}{=} \mathbb{E}\left[ \mathbb{E}\left[ \hat{G}(R_1)\textbf{1}_{\{R_1<\min{R_2,Z_2}\}}|R_1\right] \right] \\
& = \mathbb{E}\left[\hat{G}(R_1) \mathbb{E}\left[ \textbf{1}_{\{R_1<\min{R_2,Z_2}\}}|R_1\right] \right] \\
& = \mathbb{E}\left[ \hat{G}(R_1) \mathbb{P}\left( \min\{R_2,Z_2\}>R_1|R_1\right) \right],
\end{split}
\end{equation*}
where $(a)$ follows by properties of the conditional expectation. Define 
\begin{equation*}
G(r)=\hat{G}(r)\mathbb{P}\left( \min\{R_2,Z_2\}>R_1|R_1 = r\right),
\end{equation*} 
we only have left to find an explicit expression for $\mathbb{P}\left( \min\{R_2,Z_2\}>R_1|R_1=r\right)$. Because $R_1$ is independent of $(R_2,Z_2)$, 
\begin{equation*}
\mathbb{P}(\min{\{R_2,Z_2\}}>R_1|R_1=r) = \mathbb{P}(\min{\{R_2,Z_2\}}>r),
\end{equation*}
and then  
\begin{equation*}
\begin{split}
\mathbb{P}(\min{\{R_2,Z_2\}}>r) = 1-F_{R_2}(r)-F_{Z_2}(r)+F_{R_2,Z_2}(r,r),
\end{split}
\end{equation*}
where $F_{R_2}$, $F_{Z_2}$, and $F_{R_2,Z_2}$ are the CDF of $R_2$, $Z_2$, and $(R_2,Z_2)$ that can be explicitly obtained from equation \eqref{jointDensityFunction}.

In the same fashion,
\begin{equation*}
\begin{split}
\mathbb{E} & \left[\hat{H}(R_2,Z_2)\textbf{1}_{\{R_2<\min{R_1,Z_2}\}}\right] \\
& = \mathbb{E} \left[\hat{H}(R_2,Z_2)\mathbb{P}(\min\{R_1,Z_2\}>R_2|R_2,Z_2)\right], \\
\mathbb{E} & \left[ \hat{K}(R_2,Z_2)\textbf{1}_{\{Z_2<\min{R_1,R_2}\}} \right] \\
& = \mathbb{E} \left[\hat{K}(R_2,Z_2)\mathbb{P}(\min\{R_1,R_2\}>Z_2|R_2,Z_2)\right] \\
\end{split}
\end{equation*}
Define 
\begin{equation*}
\begin{split}
H(r,z):= & \hat{H}(r,z)\mathbb{P}(\min\{R_1,Z_2\}>R_2|R_2=r,Z_2=z), \\
K(r,z):= & \hat{K}(r,z)\mathbb{P}(\min\{R_1,R_2\}>Z_2|R_2=r,Z_2=z)
\end{split}
\end{equation*}
To obtain explicit formulas for $H(r,z)$ and $K(r,z)$, we proceed as before to find out that 
\begin{equation*}
\begin{split}
\mathbb{P}(\min{\{R_1,Z_2\}}>R_2|R_2=r,Z_2=z) & = (1-F_{R_1}(r))\mathbf{1}_{\{z>r\}}, \\
\mathbb{P}(\min{\{R_1,R_2\}}>Z_2|R_2=r,Z_2=z) & = (1-F_{R_1}(z))\mathbf{1}_{\{r>z\}},
\end{split}
\end{equation*}
where $F_{R_1}$ is the CDF of $R_1$. Having done this, we can evaluate the coverage probability given by equation \eqref{CP2}.

\end{document}